\documentclass[11pt]{amsart}
\usepackage{amsmath,amssymb,latexsym,dsfont}

\usepackage[small]{caption}
\usepackage{graphicx,color,mathrsfs,tikz}
\usepackage{subfigure,color}
\usepackage{cite}
\usepackage[colorlinks=true,urlcolor=blue,
citecolor=red,linkcolor=blue,linktocpage,pdfpagelabels,
bookmarksnumbered,bookmarksopen]{hyperref}
\usepackage[english]{babel}
\usepackage{units}
\usepackage{enumitem}
\usepackage[left=2.7cm,right=2.7cm,top=2.7cm,bottom=2.7cm]{geometry}
\usepackage[hyperpageref]{backref}

\usepackage{lmodern}

\usepackage{amsmath}
\usepackage{amssymb}
\usepackage{amsthm} 
\usepackage{geometry} 
\usepackage{placeins}

\usepackage{tikz}
\usetikzlibrary{patterns}
\usetikzlibrary{positioning}

\newtheoremstyle{resultat}%
   {\baselineskip}%
   {\baselineskip}%
   {\itshape}%
   {}%
   {\bfseries}%
   {}%
   { }%
   {}%
   
\theoremstyle{resultat}

\newtheorem{theorem}{Theorem}

\newtheorem{lemma}{Lemma}

\newtheorem{remark}{Remark}
\newtheorem{definition}{Definition}

\newtheoremstyle{remarque}%
   {\baselineskip}%
   {\baselineskip}%
   {}%
   {}%
   {\itshape}%
   {:}%
   { }%
   {}%

\newcommand{\N}{\mathbb{N}}

\newcommand{\R}{\mathbb{R}}
\newcommand{\C}{\mathbb{C}}

\newcommand{\loc}{loc}
\newcommand{\eps}{\varepsilon}

\newcommand{\bv}{{\bf v}}

\newcommand{\oomega}{{\bf w}}

\renewcommand{\ll}{\left\langle}
\newcommand{\rr}{\right\rangle}

\newcommand{\Hcurl}{{H}_{{curl}}(\mathbb{R}^3)}
\newcommand{\VV}{\mathcal{V}}
\newcommand{\HH}{\mathcal{H}}
\newcommand{\LL}{{L}^2(\R^3)}

\DeclareMathOperator{\imag}{{Im}}
\DeclareMathOperator{\real}{{Re}}
\DeclareMathOperator{\curl}{{curl}}

\numberwithin{equation}{section}
\numberwithin{theorem}{section}
\numberwithin{lemma}{section}
\numberwithin{proposition}{section}
\numberwithin{corollary}{section}
\numberwithin{definition}{section}
\numberwithin{remark}{section}

\author[H.-M. Nguyen]{Hoai-Minh Nguyen}
\author[V. Vinoles]{Valentin Vinoles}

\address[H.-M. Nguyen]{Department of Mathematics \newline\indent
	EPFL SB CAMA \newline\indent
	Station 8 CH-1015 Lausanne, Switzerland}
\email{hoai-minh.nguyen@epfl.ch}

\address[V. Vinoles]{Department of Mathematics \newline\indent
	EPFL SB CAMA \newline\indent
	Station 8 CH-1015 Lausanne, Switzerland}
\email{valentin.vinoles@epfl.ch}

\title[Electromagnetic wave propagation in dispersive metamaterials]{Electromagnetic wave propagation in media consisting of dispersive metamaterials}

\begin{document}

\maketitle

\begin{abstract}
We establish the well-posedness, the finite speed propagation, and a regularity result for Maxwell's equations in media consisting of dispersive (frequency dependent) metamaterials. Two typical examples for such metamaterials are materials obeying Drude's and Lorentz' models. 
The causality and the passivity  are the two main assumptions and play a crucial role in the analysis. It is worth noting that by contrast the well-posedness in the frequency domain is not ensured in general. We also provide some numerical experiments using the Drude's model to illustrate its dispersive behaviour.
\end{abstract}



\section{Introduction}
Metamaterials are smart materials engineered to have properties that have 
not yet been found in nature.
They have recently attracted a lot of attention from the scientific community, 
not only because of potentially interesting applications,
but also because of challenges in understanding their peculiar properties. 

An important class of metamaterials is   the one of negative index metamaterials (NIMs).
The study of NIMs was initiated a few decades ago in the seminal work of Veselago~\cite{Veselago},
in which the existence of such materials was postulated. The existence of NIMs was confirmed by Shelby, Smith, and Schultz in \cite{Shelby}. 
New fabrication techniques now allow the construction
of NIMs at scales that are interesting for applications, and have made them a
very  active topic of investigation. One of the interesting properties of NIMs is superlensing, 
i.e., the possibility to beat   the Rayleigh diffraction limit: no constraint between the size of the object and the wavelength is imposed. 
This was first proposed by Veselago for a slab of index $-1$ and later studied in various contexts in \cite{NicoroviciMcPhedranMilton94, PendryNegative, PendryCylindricalLenses, PendryRamakrishna,NicoroviciMcPhedranMiltonPodolskiy1}. The rigorous proof of superlensing was given in \cite{Ng-Superlensing, Ng-Superlensing-Maxwell} for related 
lens designs. 
Another interesting application of NIMs is cloaking  objects. Various schemes  were suggested  in  \cite{LaiChenZhangChanComplementary, Ng-CALR-object} and established rigorously in \cite{Ng-Negative-Cloaking, Ng-CALR-object}. NIMs can be used for cloaking sources, see,  e.g.,  \cite{MiltonNicorovici, Ng-CALR}. 
Another attracting class of metamaterials is the one of hyperbolic metamaterials (HMMs). HMMs can be used for superlensing, see  \cite{BonnetierNguyen,Jacob06, Liu07};   other promising potential applications of HMMs can be found in  \cite{Poddubny13} and references therein. The peculiar properties and the difficulties in the study of NIMs come from 
the fact that the modelling equations have sign changing coefficients. In contrast, the modelling of HMMs involves equations of changing type, elliptic in some regions, hyperbolic in others.


The well-posedness of equations modelling metamaterials has been investigated mainly in the frequency domain. Concerning NIMs,  it is now known that 
one needs to impose conditions on the coefficients of the equations near  the sign-changing coefficient-interface to insure the well-posedness,  see \cite{AnneSophieChesnelCiarlet1, CostabelErnst, Ng-Complementary, Ng-WP, Ola} and references therein, otherwise the equations are unstable, see \cite{Ng-WP}. Concerning HMMs, it is shown in \cite{BonnetierNguyen} that the stability is very sensitive with the geometry of the hyperbolic region. As far as we know, there are very few works on the stability of metamaterials apart from NIMs in the frequency domain.


This work is on Maxwell's equations in the time domain for media consisting of dispersive metamaterials. These are metamaterials whose material constants are of frequency dependence. Two typical examples for such metamaterials are the ones obeying Drude's and Lorentz' models. The study of dispersive metamaterials in the time domain for NIMs was considered by Gralak and Tip in \cite{GralakTip}. They investigated the well-posedness of Maxwell's equations in the two dimensional space setting in which NIMs occupy a half-plane and obey Drude's model.  Under the same setting, B\'{e}cache, Joly, and the second author in \cite{BJV} showed the instability of the standard PMLs 
 and design a new one in this context.  Again for this setting, the limiting amplitude principle was studied by Cassier, Hazard, and Joly in \cite{Joly16} and confirmed numerically  in \cite{vinoles}.  
 
%
%
%

In this paper, we deal with bi-anisotropic media, i.e., media for which the electric and magnetic induction fields $D$ and $B$ depend on both electric and magnetic fields $E$ and $H$. This general class of metamaterials  covers the usual anisotropic one for which $D$ (resp. $B$) depends only on $E$ (resp. $H$). In particular, the bi-anisotropic class contains NIMs  and HMMs. More precisely, we establish the well-posedess for weak solutions associated to  this model (Theorem~\ref{thm:wellPosedness} in Section~\ref{sect:math}), the finite speed propagation of weak solutions associated with these media (Theorem~\ref{thm:finitePropagationSpeed} in Section~\ref{sect:math}), and a regularity result for the weak solutions (Theorem~\ref{thm-reg} in Section~\ref{sect-reg}). By the  dispersivity, the corresponding evolution equations are non-local in time. Two key assumptions in the analysis are the causality and the passivity ones which roughly speaking say that the effect cannot precede the cause and the medium is dissipative rather than produces electromagnetic energy. In this paper, we work directly with the non-local equations. This is different from the approaches in \cite{cassier2017,GralakTip,Joly16,vinoles} where Drude's model is used and auxiliary fields are introduced to transform the non-local equations into local ones using the special structure of Drude's model. Nevertheless, this model and its particular structure are used in the simulations presented in Section ~\ref{sect:num} for simplicity. 


This paper is organized as follows. In  Section \ref{sect:model}, we  present the dispersive model for the Maxwell's equations. We there discuss bi-anisotropic media but  confine ourselves to  linear and local-in-space ones.  The well-posedness, the finite speed propagation of electromagnetic fields, and the regularity result are discussed in Section~\ref{sect:math}. Finally, some numerical experiments are presented in Section ~\ref{sect:num}.

\section{Maxwell's equations in dispersive media} \label{sect:model}
In this section, we describe Maxwell's equations in dispersive media.  
The materials presented here are mainly  from  \cite[chapter 7]{jackson}, \cite[chapters 1 and 2]{kong}, \cite[chapter IX]{landau}, \cite{nussenzveig} and \cite[chapter 1]{mackay}.
The fundamental Maxwell's equations -- without source --  are
\begin{equation}\label{eq:generalMaxwellEquations}
\left\{ \begin{aligned}
& \partial_t D(t, x) = \curl   H(t, x), \\
& \partial_t  B(t, x) =- \curl  E(t, x),
\end{aligned} \right. \quad  \mbox{ for } t \in \R,\ x \in \R^3,
\end{equation}
where $ E \in \R^3$ (resp. $ H \in \R^3$) is the electric (resp. magnetic) field and $ D\in \R^3$ (resp. $ B\in \R^3$) is the electric (resp. magnetic) induction field.  In order to close the system \eqref{eq:generalMaxwellEquations}, one adds constitutive relations that express $D$ and $B$ as functions of $E$ and $H$. For dispersive media, these relations are more conveniently presented in the frequency domain. 
In this paper, for a time-dependent field $X(t, x)$, its temporal Fourier transform  is given by
\begin{equation}\label{FT}
\widehat{ X}(\omega, x):= \frac{1}{\sqrt{2 \pi}}\int_{\R} X(t, x) e^{i\omega t}\, d t,\quad  \mbox{ for } \omega \in \R,\  x \in \R^3.
\end{equation}
In the frequency domain, the Maxwell's equations \eqref{eq:generalMaxwellEquations} are of the form
\begin{equation}\label{eq:generalMaxwellEquations2}
\left\{ \begin{aligned}
& - i \omega \widehat{D}(\omega, x) =  \curl   \widehat{H}(\omega, x), \\
& - i \omega \widehat{B}(\omega, x) =  -\curl   \widehat{E}(\omega, x),
\end{aligned} \right. \quad  \mbox{ for } \omega \in \R,\  x \in \R^3.
\end{equation}
In this paper, we consider linear bi-anisotropic materials, i.e., $D$ and $B$ depend linearly on \emph{both} $E$ and $H$ (see, e.g., \cite[Chapters 1 and 2]{kong} and \cite[Chapter 1]{mackay}). This class of materials contains  the anisotropic ones for which $D$ (resp. $B$) depends only on $E$ (resp. $B$), see, e.g., \cite[chapter 7]{jackson} and \cite[chapter IX]{landau}. We also assume that
the media considered are local in space.  The constitutive relations in the frequency domain of bi-anisotropic media are then of the form
\begin{equation}\label{eq:constitutiveRelation}
\left\{ \begin{aligned}
&  \widehat{ D}(\omega, x) 
=  \big( \eps_{rel}(x) + \widehat{\chi_{ee}} (\omega,x)\big) \widehat{ E}(\omega, x) +  \widehat{\chi_{em}} (\omega,x)  \widehat{ H}(\omega, x),  \\
&  \widehat{ B}(\omega, x)
=   \widehat{\chi_{me}} (\omega,x) \widehat{ E}(\omega, x) + \big( \mu_{rel}(x) + \widehat{\chi_{mm}} (\omega,x)  \widehat{ H}(\omega, x), 
\end{aligned}\right.
 \quad  \mbox{ for } \omega \in \R,\  x \in \R^3. 
\end{equation}
Here $\widehat{\chi_{ij}}(\omega,x)$, $(i,j) \in \{e,m\}^2$, are $3 \times 3$ matrices called the \emph{susceptibilities} that characterize the dispersive effects of the medium, i.e.,  its response with respect to the frequency $\omega$ at the point $x$. The permittivity $\eps$ and the permeability $\mu$ of the medium are given  by
\begin{equation}
\widehat{\eps} : = \eps_{rel} + \widehat \chi_{ee} \quad \mbox{ and } \quad \widehat{\mu}: = \mu_{rel} + \widehat \chi_{mm}.
\end{equation}
We assume that 
\begin{equation}\label{epsmu}
\mbox{$\eps_{rel}$ and $\mu_{rel}$ are two  $3 \times 3$ real symmetric uniformly elliptic matrices defined in $\R^3$.}
\end{equation}
One can check that $\eps_{rel}$ and $\mu_{rel}$ correspond  respectively to $\widehat \eps$ and $\widehat \mu$ for large frequency provided that $\widehat{\chi_{ee}}$ and $\widehat{\chi_{mm}}$ are in $L^1(\R, L^\infty(\R^3)^{3 \times 3})$.  These constitution relations are Lorentz covariants (see, e.g.. \cite[chapter 2]{kong}). 

If all the $\widehat{\chi_{ij}}$ are independent of $\omega$, the corresponding medium is called a \emph{dielectric medium};  otherwise it is a \emph{dispersive medium}. In the case $\chi_{em}  = \chi_{me} = 0$, \eqref{eq:constitutiveRelation} models anisotropic media. In a special case of \eqref{eq:constitutiveRelation} in which $\chi_{ij}$ are isotropic, media are called reciprocal  chiral and  consist of Pasteur  and Tellegen ones, see, e.g.,  \cite{shivola}.  


Set 
\begin{equation}\label{def-lambda}
\widehat{\lambda_{ij}}(\omega, x) := - i \omega \widehat{\chi_{ij}}(\omega, x), \quad \mbox{ for } (i,j) \in \{e,m\}^2,\ \omega \in \R,\  x \in \R^3.
\end{equation}
Inserting \eqref{eq:constitutiveRelation} in \eqref{eq:generalMaxwellEquations2} gives, for  $\omega \in \R$ and  $x \in \R^3$, 
\begin{equation}\label{eq:generalMaxwellEquations3}
\left\{ \begin{aligned}
& - i \omega \eps_{rel}(x) \widehat{E}(\omega, x) + \widehat{\lambda_{ee}} (\omega,x) \widehat{ E}(\omega, x) +  \widehat{\lambda_{em}} (\omega,x)  \widehat{ H}(\omega, x) =  \curl   \widehat{H}(\omega, x), \\
& - i \omega  \mu_{rel} (x) \widehat{H}(\omega, x) + \widehat{\lambda_{me}} (\omega,x)  \widehat{ E}(\omega, x) +  \widehat{\lambda_{mm}} (\omega,x)  \widehat{ H}(\omega, x) =-  \curl   \widehat{E}(\omega, x). 
\end{aligned} \right.  
\end{equation}
One can derive that $\lambda_{ij}$ is analytic in the upper half $\omega$-plane and continuous up to the boundary of the half plane as long as 
\begin{equation}
\lambda_{ij} \in L^1\big(\R, L^\infty( \R^3)^{3 \times 3} \big) + L^\infty\big(\R, L^\infty( \R^3)^{3 \times 3} \big), \qquad \mbox{ for } (i,j) \in \{e,m\}^2. 
\end{equation}
This allows to use Cauchy's theorem and obtain a relation between $\real \widehat{\chi_{ij}}$ and $\imag \widehat{\chi_{ij}}$, which is known as  the Kramers-Kronig relation (see, e.g., \cite{nussenzveig, Toll} for further information). 
We will make the following assumptions on $\lambda_{i j}$: 
\begin{equation}\label{assumption-lambda}
\widehat{\lambda_{i j}},  \, \lambda_{i j} \in L^1_{\loc}\big(\R, L^\infty( \R^3)^{3 \times 3} \big) \ \mbox{ and } \ \lambda_{ij}  \mbox{ is real-valued},  \qquad \mbox{ for } (i, j) \in \big\{e, m \big\}^2. 
\end{equation}
By the inverse Fourier transform 
\begin{equation}
X(t, x) =  \frac{1}{\sqrt{2 \pi}}\int_{\R} \widehat{ X}(\omega, x) e^{-i \omega t}\,   d \omega,  \quad  \mbox{ for } t \in \R,\ x \in \R^3, 
\end{equation}
the corresponding system of \eqref{eq:generalMaxwellEquations3} in time domain is  \begin{equation}\label{eq:generalMaxwellEquations4}
\left\{ \begin{aligned}
& \eps_{rel}(x) \partial_t E(t, x) + (\lambda_{ee} * E)(t, x) +  (\lambda_{em} *H)(t, x) =  \curl H(t, x), \\
&  \mu_{rel}(x) \partial_t H(t, x)+  (\lambda_{me}* E)(t, x) + (\lambda_{mm}*H)(t,x)   = -\curl E(t, x),
\end{aligned}\right.
 \quad  t \in \R,\  x \in \R^3,
\end{equation}
where $*$ stands for the convolution with respect to time $t$.

\medskip 
Two fundamental assumptions physically relevant on the model,  causality and passivity,  are imposed.  

\medskip 
\noindent {\bf Causality}: the effect
cannot precede the cause, i.e., the present states of the system depend only on its states in the past. Mathematically, one requires 
\begin{equation}\label{eq:causality}
\lambda_{ij}(t) = 0, \qquad \text{for all $t<0$ and for all $(i,j) \in \{e,m\}^2$}.
\end{equation}

Under this assumption, we have, for $(i, j)  \in \{e, m \}^2$, 
\begin{equation}\label{eq-causality1}
(\lambda_{ij} *  X)(t,\cdot) = \int_{-\infty}^t \lambda_{ij} (t-\tau,\cdot) X(\tau,\cdot)\,  d \tau = \int_{0}^\infty \lambda_{ij} (\tau,\cdot) X(t - \tau,\cdot)\,  d \tau,\quad  \mbox{ for } t \in \R. 
\end{equation}



\medskip 
\noindent {\bf Passivity}: One assumes,  for almost every $x \in \R^3$, for almost every $\omega \in \R$ and for all $X \in \C^6$, that\footnote{Here $\cdot$ stands for the Euclidean scalar product in $\C^6$.}
\begin{equation}\label{eq:passivity}
\real \left( \begin{bmatrix}
\widehat{\lambda_{ee}}(\omega,x) & \widehat{\lambda_{em}}(\omega,x) \\ 
\widehat{\lambda_{me}}(\omega,x) & \widehat{\lambda_{mm}}(\omega,x)
\end{bmatrix} X \cdot \overline X \right) \ge 0, 
\end{equation}

Under the terms of $\chi_{ij}$ (see \eqref{def-lambda}),  condition~\eqref{eq:passivity} can be written as 
\begin{equation}\label{eq:passivity-1}
\omega \imag \left( \begin{bmatrix}
\widehat{\chi_{ee}}(\omega,x) & \widehat{\chi_{em}}(\omega,x) \\ 
\widehat{\chi_{me}}(\omega,x) & \widehat{\chi_{mm}}(\omega,x)
\end{bmatrix} X \cdot \overline X \right) \ge 0. 
\end{equation}
Assumption \eqref{eq:passivity} means that the medium is dissipative, i.e., it 
does not produce electromagnetic energy by itself.  We emphasize that  no assumption on the sign of the real part of the $\chi_{ij}$ in  \eqref{eq:passivity-1}
is required (or equivalently on the imaginary part of the $\lambda_{ij}$ in \eqref{eq:passivity}). Moreover, no symmetry on the $\chi_{ij}$ (or equivalently on the $\lambda_{ij}$) is assumed.  

\medskip 
Some comments on these assumptions are in order in the anisotropic case ($\lambda_{em} = \lambda_{me}  = 0$) and in the frequency domain. It is possible for some frequencies that $\widehat \eps$ and $\widehat \mu$ are both negative in some regions.
This corresponds to NIMs (see Lorentz' and Drude's models below). It is also possible that $\widehat \eps$ and $\widehat \mu$ have both positive and negative eigenvalues in some region. In this case, one deals with HMMs.  In the anisotropic case, condition~\eqref{eq:passivity-1} is equivalent to\footnote{Here for a $3 \times 3$ matrix $A$, we denote $A \le 0$ if $A x \cdot x \le 0$ for all $x \in \R^3$.}
\begin{equation}\label{sign-eps-mu}
\omega \imag \widehat \eps(\omega), \ \omega \imag \widehat \mu(\omega)  \ge 0,\quad  \mbox{ for almost all } \omega \in \R. 
\end{equation}
Condition~\eqref{sign-eps-mu} ensures that when small loss is added, the problem associated with the outgoing (Silver-M\"uller) condition at infinity is well-posed (see,  e.g., \cite{Ng-Superlensing-Maxwell}). Adding a small loss is the standard mechanism to study phenomena related to metamaterials in the frequency domain. 
Nevertheless, condition~\eqref{sign-eps-mu} does not exclude the ill-posedness in the frequency domain (see \cite[Proposition 2]{Ng-WP}). As one sees later, even if the problem is ill-posed in the frequency domain for some frequency, the well-posedness is ensured for the problem in the time domain under  roughly speaking the causality and passivity conditions mentioned above (see Theorem~\ref{thm:wellPosedness}). 




\medskip 
We next recall two typical examples of dispersive anisotropic media ($\chi_{me}  = \chi_{em} = 0$) satisfying condition \eqref{assumption-lambda},  the causality \eqref{eq:causality} and the passivity \eqref{eq:passivity}. 
The first one is media obeying \emph{Lorentz' model}. For a homogeneous isotropic medium, the susceptibilities $\chi_{ee}$ and $\chi_{mm}$ are of the form (see e.g.,  \cite[(7.51)]{jackson})
\begin{equation}\label{eq:LorentzModelFrequencyDomain}
\widehat{\chi}(\omega) = \sum_{\ell=1}^n \frac{\omega_{p,\ell}^2}{\omega_{0,\ell}^2 - \omega^2  - 2i \gamma_\ell \omega}\, {I_3},\quad  \mbox{ for } \omega \in \R. 
\end{equation}
where $\omega_{p,\ell}$ (resp. $\omega_{0,\ell}$ and $\gamma_\ell$) are positive (resp. non negative)  material constants. Here and in what follows ${I_3}$ denotes the $3 \times 3$ identity matrix.  Using the residue theorem, one can  show (see e.g.,  \cite[(7.110)]{jackson}) that for $t \in \R$ one has
\begin{equation}\label{eq:LorentzModelTimeDomain}
\chi(t) =  \sqrt{2 \pi} \theta(t) \sum_{\ell=1}^n \omega_{p,\ell}^2\, \frac{\sin(\nu_\ell t)}{\nu_\ell}\, e^{-\gamma_\ell t }\, {I_3}
\quad \text{and} \quad \lambda(t) = \sqrt{2 \pi} \theta(t) \sum_{\ell=1}^n \omega_{p,\ell}^2\, \frac{d}{dt} \left( \frac{\sin(\nu_\ell t)}{\nu_\ell}\, e^{-\gamma_\ell t } \right) {I_3},
\end{equation}
where $\nu_\ell^2 = \omega_{0,\ell}^2-\gamma_\ell^2$ and  $\theta$ is the Heaviside function, i.e., $\theta(t) = 1$ if $t \ge 0$ and $\theta(t) = 0$ otherwise. Here $\lambda$ is defined in such a way that $\widehat{\lambda}(\omega) = -  i \omega \widehat{\chi}(\omega)$ for $\omega \in \R$.

 One can easily check that Lorentz' model satisfies conditions \eqref{assumption-lambda},  \eqref{eq:causality}, and \eqref{eq:passivity} (which implies  \eqref{eq:passivity-1}). 

The second example is  \emph{Drude's model}. It is  a particular case of the Lorentz model \eqref{eq:LorentzModelFrequencyDomain} with $n=1$ and $\omega_{0,1} = 0$:
\begin{equation}\label{eq:DrudeModel}
\widehat{\chi}(\omega) = \frac{\omega_p^2}{- \omega^2  - 2i \gamma \omega}\, {I_3},\quad  \mbox{ for } \omega \in \R. 
\end{equation}
One thus has
\begin{equation}\label{eq:DrudeModelTimeDomain}
\chi(t) =  \sqrt{2 \pi}  \omega_p^2 \gamma^{-1} (1-e^{-\gamma t}) \theta(t)\, {I_3}\quad \text{and} \quad \lambda(t) =  \sqrt{2 \pi}  \omega_p^2 e^{-\gamma t} \theta(t)\, {I_3},\quad  \mbox{ for } t \in \R.  
\end{equation}

\begin{remark}\label{rem-Milton} \rm Using Lorentz' model for $\widehat \mu$ is probably not too realistic (see, e.g., \cite[\S 60]{landau}) but has the advantage that the imaginary part of $\widehat{\chi_{mm}}$ has a minimum which can be minimized to weaken the loss effect. 
\end{remark}

\begin{remark} \rm
Using homogeneization theory, one can obtain HMMs from positive index materials and NIMs (see, e.g., \cite{BonnetierNguyen}).
\end{remark}

\begin{remark}  \rm In \cite{cassier2017}, the authors proposed various conditions on  dispersive models. 
Some of their postulates are not required here. 


\end{remark}

\section{Electromagnetic wave propagation in dispersive media}\label{sect:math}


In this paper, we study \eqref{eq:generalMaxwellEquations4} under the form of the initial problem at the time $t = 0 $ assuming that the data are known in the past $t<0$. 
Set 
\begin{equation}
\label{eq:truncatedConvolutionProduct}
(\lambda_{ij} \star  X)(t,\cdot) := \int_{0}^t \lambda (t-\tau,\cdot) X(\tau,\cdot)\,   d \tau, \quad  \mbox{ for } t > 0. 
\end{equation}
For $X = E$ or $H$, under the causality assumption \eqref{eq:causality}-\eqref{eq-causality1}, one has for $t>0$ that
\begin{equation}\label{reformulate}
\begin{aligned}
(\lambda_{ij} *  X)(t,\cdot)&  =  \int_{0}^t \lambda_{ij} (t-\tau,\cdot) X(\tau,\cdot)\,  d \tau + \int_{-\infty}^0 \lambda_{ij} (t-\tau,\cdot) X(\tau,\cdot)\,  d \tau  \nonumber \\
& = (\lambda_{ij} \star  X)(t,\cdot) + \int_{-\infty}^0 \lambda_{ij} (t-\tau,\cdot) X(\tau,\cdot)\,  d \tau.
\end{aligned}
\end{equation}
Hence if  the data are known for the past $t < 0$, then the last term is known at time $t > 0$. 
With the presence of sources,  one can then reformulate system   \eqref{eq:generalMaxwellEquations4} under the form
\begin{equation}\label{eq:problem0}
\left\{ 
\begin{aligned}
& \eps_{rel}(x) \partial_t E(t, x) + (\lambda_{ee} \star E)(t, x) +  (\lambda_{em} \star H)(t, x) =  \curl H(t, x)+f_e(t,x), \\
&  \mu_{rel}(x)  \partial_t H(t, x)+  (\lambda_{me}\star E)(t, x) + (\lambda_{mm}\star H)(t,x)   = - \curl E(t, x)+f_m(t,x),\\
& E(0,x) = E_0(x),  \ H(0,x) = H_0(x), 
\end{aligned}\right.
\end{equation}
for $t > 0$ and $x \in \R^3$. Here $E_0, \, H_0$ are the initial data at time $t =0$ and $f_e, \ f_m$ are given fields which can be considered as ``effective" sources since they also take into account the last terms in \eqref{reformulate}. Note that if sources are 0 for $t < 0$, then the initial problem considered here with $E_0  = H_0 = 0$ gives exactly the solutions of \eqref{eq:generalMaxwellEquations4} admitting that $E = H = 0$ for $t < 0$ since there is no source for $t < 0$ (see Remark~\ref{rem-Unique}).  

Set
\begin{equation}\label{eq:notations}
u := \begin{bmatrix}
E \\ H
\end{bmatrix}, \quad 
u_0 := \begin{bmatrix}
E_0 \\ H_0
\end{bmatrix}, \quad
f := \begin{bmatrix}
f_e \\ f_m
\end{bmatrix}, \quad
\mathbb A u := \begin{bmatrix}
\curl H \\ - \curl E
\end{bmatrix},
\end{equation}
\begin{equation}\label{eq:notations-1}
\Lambda := \begin{bmatrix}
\lambda_{ee} & \lambda_{em} \\
\lambda_{me} & \lambda_{mm} 
\end{bmatrix} \quad \mbox{ and } \quad 
M := \begin{bmatrix}
\eps_{rel} & 0 \\
0 & \mu_{rel} 
\end{bmatrix}.
\end{equation}
System \eqref{eq:problem0} can then be rewritten in the following compact form: 
\begin{equation}\label{eq:problem}
\left\{ 
\begin{aligned}
& M(x) \partial_t u(t,x) + (\Lambda \star u)(t,x) = \mathbb A u(t,x) + f(t,x), \\
& u(0,x) = u_0(x),  \\
\end{aligned}\right.  \quad  \mbox{ for } t> 0,\  x \in \R^3.
\end{equation}
The goal of this paper  is to establish the well-posedness, the finite speed propagation and to present a regularity result for \eqref{eq:problem}.



Define
\begin{equation}
\HH := \LL^3 \times \LL^3 \quad \mbox{ and } \quad \VV := \Hcurl \times \Hcurl, 
\end{equation}
equipped with the standard inner products induced from $\LL^3$ and $\Hcurl$.  One can verify that $\HH$ and $\VV$ are Hilbert spaces. 
We also denote 
\begin{equation}\label{def-M6}
\mbox{$\mathcal M_6( L^\infty(\R^3))$ the space of $6 \times 6$ real matrices whose entries are $ L^\infty(\R^3)$ functions.}
\end{equation}
In what follows, in the time domain, we only consider {\it real} quantities. 

\medskip
The first result of this paper is the well-posedness of \eqref{eq:problem}, whose proof is given in Section~\ref{sec:wp}:

\begin{theorem}\label{thm:wellPosedness}
Let $T \in (0,+\infty)$, $u_0  \in \HH$, $f \in {L}^1(0,T;\HH)$,  and $\Lambda \in  L^1\big(0,T;\mathcal M_6( L^\infty(\R^3) \big)$. Assume that \eqref{epsmu}, \eqref{assumption-lambda}, \eqref{eq:causality} and
\eqref{eq:passivity} hold.  There exists a unique weak solution $u \in  L^\infty(0,T;\HH)$  of \eqref{eq:problem} on $(0,T)$. Moreover, the following estimate holds 
\begin{equation}\label{eq:APrioriEstimate}
\ll M u(t, \cdot), u(t, \cdot) \rr_\HH  \le  \left(\ll M u_0, u_0 \rr_\HH^{1/2}   + C \int_0^t \|  f(s, \cdot) \|_{\HH} \,   ds\right)^2 \quad \mbox{ in } (0, T), 
\end{equation}
where $C$ is a positive constant depending only on the coercivity of $M$. 
\end{theorem}

%
%


The notion of weak solutions for \eqref{eq:problem} is: 

\begin{definition}\label{defi:weakSolution} \rm
Let $T \in (0,+\infty)$, $u_0 \in \HH$ and $f \in  L^1(0,T;\HH)$. A function  $u \in  L^\infty(0,T;\HH)$ is called a \emph{weak solution} of \eqref{eq:problem} on $[0,T]$ if 
\begin{equation}\label{eq:defWeakSolution}
\frac{ d}{ dt} \ll M u(t, \cdot ),v \rr_\HH + \ll (\Lambda \star u)(t, \cdot),v\rr_\HH   =  \ll u(t, \cdot), \mathbb A v \rr_\HH  + \ll f(t, \cdot),v \rr_\HH \mbox{ in } (0, T)  \mbox{ for all $v \in \VV$},
\end{equation}
and \begin{equation}\label{eq:initialConditions}
u(0, \cdot) = u_0. 
\end{equation}
\end{definition}


\begin{remark} \rm One can easily check that if $u$ is a smooth solution and decays enough at infinity, then $u$ is a weak solution by integration by parts, and that if $u$ is a weak solution and smooth then $u$ is a classical solution. 
\end{remark} 

Some comments on Definition~\eqref{defi:weakSolution} are in order. Equation \eqref{eq:defWeakSolution} is understood in the distributional sense. Initial condition   \eqref{eq:initialConditions}  is understood as 
\begin{equation}\label{trace-sense}
\ll M u(0, \cdot),v \rr_\HH  = \ll M u_0,v \rr_\HH, \quad \mbox{ for all } v \in \VV. 
\end{equation}
Under the assumptions $u \in  L^\infty(0,T;\HH)$, $v \in \VV$, $f \in {L}^1(0,T;\HH)$  and $\Lambda \in  L^1 \big(0,T;\mathcal M_6( L^\infty(\R) \big)$, one can check that $ \ll (\Lambda \star u)(t),v\rr_\HH$, $ \ll u(t), \mathbb A v \rr_\HH$, $\ll f(t),v \rr_\HH$ are in $L^1(0, T)$. It follows from \eqref{eq:defWeakSolution} that 
\begin{equation}
\ll M u(t),v \rr_\HH  \in W^{1, 1}(0, T). 
\end{equation}
This in turn ensures the trace sense of $\ll M u(0, \cdot),v \rr_\HH $ in \eqref{trace-sense}. 

\medskip 

We next discuss the finite speed propagation for \eqref{eq:problem}. 
In what follows, $B(a,R)$ stands for the ball of $\R^3$ of radius $R>0$ centred at $a \in \R^3$ and $\partial B(a, R)$ denotes its boundary. In the case $a = 0$ -- the origin -- we simply denote $B(0, R)$ by $B_R$. 
Set  
\begin{equation}
c(x) :=  \gamma_e (x) \gamma_{m} (x),   \quad  \mbox{ for } x \in \R^3, 
\end{equation}
where $\gamma_e(x)$ and $\gamma_m(x)$ are respectively the largest eigenvalues of $\eps_{rel}(x)^{-1/2}$ and $\mu_{rel}(x)^{-1/2}$. According to assumptions \eqref{epsmu}, $c(x)$ is bounded below and above by a positive constant. For $a \in \R^3$ and $R> 0$, we denote  
\begin{equation}\label{eq:defSpeed}
c_{a, R}:= \mathop{\mbox{ess sup}}_{x \in B(a, R)} c(x).   
\end{equation}
The second  result of this paper is on the  finite speed propagation of \eqref{eq:problem}, whose proof is given in Section~\ref{sec:fps}. 

\begin{theorem}\label{thm:finitePropagationSpeed} Let $R>0$, $a \in \R^3$ and $u_0  \in \HH$. For $T > R/c_{a,R}$, let $f \in  L^1(0, T;\HH)$ and $\Lambda \in  L^1(0, T;\mathcal M_6\big( L^\infty(\R^3)\big)$. Assume that \eqref{epsmu}, \eqref{assumption-lambda}, \eqref{eq:causality} and
\eqref{eq:passivity} hold, 
\begin{equation}
{supp}\, u_0 \cap B(a,R)  = \emptyset,
\end{equation}
and
\begin{equation}
{supp}\, f(t, \cdot) \cap B(a,R-c_{a, R} t) = \emptyset, \quad \mbox{ for almost every } t \in (0, R/ c_{a, R}). 
\end{equation}
Let $u \in  L^\infty(0, T; \HH)$ be  the unique weak solution of \eqref{eq:problem} on $(0, T)$. Then 
\begin{equation}
{supp}\, u(t, \cdot) \cap  B(a,R-c_{a, R} t) = \emptyset, \quad \mbox{ for almost every } t \in (0, R/ c_{a, R}). 
\end{equation}
\end{theorem}

We finally discuss  the regularity of the weak solutions of \eqref{eq:problem}. To motivate the regularity result stated below, let  us first assume that $u$ is a weak solution of  \eqref{eq:problem}  and that $u$, $\Lambda$, and $f$ are regular in $[0, T] \times \R^3$.  Set
\begin{equation}
v(t, x):= \partial_t u(t, x), \quad  \mbox{ for }  t \in (0,T),\ x \in  \R^3.
\end{equation}
Differentiating  \eqref{eq:problem} with respect to $t$, we have  
\begin{equation}
M(x) \partial_t v (t, x) + (\Lambda \star v) (t, x) = \mathbb A v(t, x) + g(t, x), \quad  \mbox{ for }  t \in (0,T),\ x \in  \R^3, 
\end{equation}
where 
\begin{equation}
g(t, x) := \partial_t f(t, x) - \Lambda(t, x) u_0(x)  \quad  \mbox{ in } [0, T) \times \R^3. 
\end{equation}
Applying Theorem~\ref{thm:wellPosedness} to $v$ and noting that 
$M v(0, \cdot) =  \mathbb A u_0 + f(0, \cdot)$, we obtain  
\begin{equation}
\| v(t, \cdot) \|_\HH \le C \Big( \| u_0 \|_\VV + \| f(0, \cdot) \|_\HH + \int_0^t \|\partial_s f (s, \cdot) \|_{\HH} + \| \Lambda(s, \cdot)\|_{L^\infty(\R^3)} \, ds  \Big), \quad  \mbox{ in }  (0,T),
\end{equation}
for some positive constant $C$ depending only on the ellipticity of $M$.

\medskip 
In fact, we can prove  the following result.
 \begin{theorem}\label{thm-reg}
Let $T \in (0,+\infty)$, $u_0  \in \VV$, $f \in {L}^1(0,T;\HH)$,  and $\Lambda \in  L^1\big(0,T;\mathcal M_6( L^\infty(\R^3) \big)$. Assume that \eqref{epsmu}, \eqref{assumption-lambda}, \eqref{eq:causality} and
\eqref{eq:passivity} hold and $\partial_t f \in {L}^1(0,T;\HH)$. Let  $u \in  L^\infty(0,T;\HH)$ be the unique weak solution  of \eqref{eq:problem} on $(0,T)$. Then $\partial_t u \in L^\infty(0, T;\HH)$ and,  for $t \in (0, T)$, 
\begin{equation}\label{reg-1}
\| \partial_t u(t, \cdot)\|_\HH^2  \le  C \left(  \| u_0 \|_\VV + \| f(0, \cdot) \|_\HH + \int_0^t \|\partial_s f (s, \cdot) \|_{\HH} + \| \Lambda(s, \cdot)\|_{L^\infty(\R^3)} \| u(s, \cdot) \|_{\HH} \, ds \right)^2,  
\end{equation}
for some positive constant  $C$ depending only on the coercivity of $M$.
\end{theorem}

\begin{remark} \rm One can bound $\| u(s, \cdot) \|_{\HH}$ using \eqref{eq:APrioriEstimate}. 

\end{remark}

The next three sections are respectively devoted to the proof of Theorems \ref{thm:wellPosedness}, \ref{thm:finitePropagationSpeed}, and \ref{thm-reg}.

\subsection{Proof of Theorem \ref{thm:wellPosedness}}\label{sec:wp}

%
%

The proof is based on the standard Galerkin approach. We first establish the existence of a weak solution.  Let $(\phi_k)_{k  \in \N}$ be a (real) orthogonal  basis of $\VV$. For $n \in \N$, consider $u_n$ of the form  
\begin{equation}\label{eq:defun}
u_n(t, x) = \sum_{k=1}^n d_{n,k}(t) \phi_k(x),\quad  \mbox{ for }  t \in (0,T),\ x \in \R^3,
\end{equation}
such that for all $k \in \{1,\ldots,n\}$
\begin{equation}\label{eq:approximateWeakSolution}
 \frac{ d}{ dt} \ll M u_n(t),\phi_k \rr_\HH + \ll (\Lambda \star u_n)(t),\phi_k \rr_\HH = \ll u_n(t), \mathbb A\, \phi_k \rr_\HH +\ll f(t),\phi_k \rr_\HH,\quad  \mbox{ in }  (0,T),
\end{equation}
and
\begin{equation}\label{projection-u0}
u_{n}(0) =  u_{0,n}, \mbox{ the projection of $u_0$ to the space spanned by $\big\{\phi_1, \cdots, \phi_n\big\}$ in $\HH$}.
\end{equation}
Since $(\phi_k)_{k \in \N}$ is linearly independent in $\VV$, it is also linearly independent in $\HH$. This implies that the $n \times n$ matrix whose $(i,j)$-entry is given by $\langle \phi_i, \phi_j \rangle_\HH$ is invertible. 
Since 
\begin{equation}\label{on-Lambda}
\| \Lambda \star u \|_{ L^\infty(0,T;\HH)} \le \| \Lambda \|_{ L^1\big(0,T;\mathcal M_6( L^\infty(\R) \big)} \| u \|_{{L}^\infty(0,T;\HH)}, 
\end{equation}
the existence and uniqueness of $u_n$  follow by a standard point-fixed argument (see, e.g. \cite[Theorem 2.1.1]{Vol-Burton}).


\medskip 
We now derive an estimate for $u_n$. The key point of the analysis is the following two observations : 
\begin{equation}\label{eq:weakPassivity}
 \int_0^t  \ll (\Lambda \star v)(s, \cdot),  v(s, \cdot) \rr_\HH  d s \ge 0 \quad \mbox{ for } v \in L^\infty(0, T; \HH), \ t \in (0,T),
\end{equation}
and 
\begin{equation}\label{integration-Auu}
\ll v,  \mathbb A \,v \rr_\HH  = 0, \quad   \mbox{ for } v \in \VV. 
\end{equation}
Note that \eqref{integration-Auu} follows by an integration by parts and the density of $\mathcal C^1_{c}(\R^3)^6$ in $\VV$.  We now verify \eqref{eq:weakPassivity}.  Let  $\bv$ be the extension of $v$ in $\R$ by 0 for  $t \in \R \setminus [0, T]$. It follows from  \eqref{eq:causality} and \eqref{eq:truncatedConvolutionProduct} that
\begin{equation}
(\Lambda \star v)(s, \cdot) = (\Lambda * \bv)(s, \cdot),\quad \mbox{ for } s \in [0,t]. 
\end{equation}
By Parseval's identity,  one has, for $t \in (0,T)$,
\begin{equation}
\begin{aligned}
\int_0^t  \ll (\Lambda\star v)(s, \cdot ), v(s, \cdot) \rr_{\HH}   ds
& =  \int_{\R}  \ll  (\Lambda *  \bv)(s, \cdot ), \bv(s, \cdot) \rr_{\HH}   ds  \\
& = \real \int_\R  \ll \mathcal F (\Lambda * \bv) (\omega, \cdot) ,\overline{\widehat{\bv}(\omega, \cdot)} \rr_{\HH}   d\omega\\
& =  \int_\R \real \ll \widehat{\Lambda}(\omega, \cdot) \widehat{\bv}(\omega, \cdot) , \overline{\widehat{\bv}(\omega, \cdot)}\rr_{\HH}   d\omega \ge 0,
\end{aligned}
\end{equation}
thanks to the passivity \eqref{eq:passivity}. Assertions \eqref{eq:weakPassivity} and \eqref{integration-Auu} are proved. 

\medskip 
Multiplying  \eqref{eq:approximateWeakSolution} by $d_{n, k}(t)$ and summing with respect to $k$ yields that, in $(0,T)$,
\begin{multline}\label{eq:tmp1}
\frac{1}{2} \frac{d}{dt} \ll M u_n(t, \cdot),\ u_n(t, \cdot) \rr_\HH  + \ll (\Lambda \star u_n)(t, \cdot),u_n(t, \cdot)\rr_\HH \\
= \ll u_n(t, \cdot), \mathbb A  u_n(t, \cdot) \rr_\HH +\ll f(t, \cdot),u_n(t, \cdot) \rr_\HH.
\end{multline}
Integrating \eqref{eq:tmp1} from $0$ to $t$  and using \eqref{integration-Auu}, we obtain that, in $(0,T)$, 
\begin{multline}\label{eq:tmp2}
\frac{1}{2}\ll M u_n(t, \cdot), u_n(t, \cdot) \rr_\HH 
+ \int_0^t  \ll (\Lambda\star u_n)(s, \cdot),u_n(s, \cdot) \rr_{\HH}  ds \\
= \frac{1}{2}\ll M u_n(0, \cdot), u_n(0, \cdot) \rr_\HH  + \int_0^t \ll  f(s, \cdot),u_n(s, \cdot) \rr_{\HH}   ds .
\end{multline}
We derive from  \eqref{eq:weakPassivity}  that, in $(0,T)$,
\begin{equation}\label{eq:tmp3}
\ll M u_n(t, \cdot), u_n(t, \cdot) \rr_\HH 
\le \ll M u_{0, n}, u_{0, n} \rr_\HH  + 2 \int_0^t \|  f(s, \cdot) \|_{\HH} \| u_{n}(s, \cdot) \|_{\HH}\,   ds.
\end{equation}
By Gr\"onwall's inequality (see Lemma~\ref{lem-GW} below) and assumptions \eqref{epsmu}, one gets from \eqref{eq:tmp3}
\begin{equation}\label{eq:approximateAPrioriEstimate}
\ll M u_n(t, \cdot), u_n(t, \cdot) \rr_\HH \le  \left(\ll M u_{n, 0}, u_{n, 0} \rr_\HH^{1/2} + C \int_0^T \|  f(s) \|_{\HH} \,   ds\right)^2, \mbox{ in } (0,T),
\end{equation}
where $C$ is a positive constant depending only on the ellipticity of $M$.  Since $\| u_{n, 0} \|_\HH \le \| u_0 \|_\HH$ by \eqref{projection-u0}, the sequence $(u_n)_{n \in \N}$ is hence  bounded in ${L}^\infty(0,T;\HH)$. Up to a subsequence, $(u_n)_{n \in \N}$ weakly star converges to $u \in {L}^\infty(0,T;\HH)$. It is clear  from \eqref{eq:approximateAPrioriEstimate} that  \eqref{eq:APrioriEstimate} holds and, for $k \in \N$,  
\begin{equation}\label{eq:tmptmptmp2}
 \frac{ d}{ dt} \ll M u(t, \cdot),\phi_k \rr_\HH + \ll (\Lambda \star u )(t, \cdot),\phi_k \rr_\HH = \ll u(t, \cdot),  \mathbb A\,\phi_k \rr_\HH +\ll f(t, \cdot),\phi_k \rr_\HH, \quad \mbox{ in } (0, T).
\end{equation} 
Since $(\phi_k)$ is dense in $\VV$, we derive that for $\phi \in \VV$
\begin{equation}\label{eq:tmptmptmp3}
 \frac{ d}{ dt} \ll M u(t, \cdot),\phi \rr_\HH + \ll (\Lambda \star u )(t, \cdot),\phi \rr_\HH = \ll u(t, \cdot), \mathbb A\, \phi \rr_\HH +\ll f(t, \cdot),\phi \rr_\HH,  \quad \mbox{ in } (0, T).
\end{equation}
One can also check that the initial condition \eqref{eq:initialConditions} holds.

\medskip 
We finally establish the uniqueness of $u$.  It suffices  to show that if $u \in  L^\infty(0,T;\HH)$ is a weak solution of \eqref{eq:problem} on $[0,T]$ with $u_0 = 0$ and $f = 0$, then $u = 0$. Set
\begin{equation}
U(t, x) := \int_0^t u(s, x)\, ds, \quad \mbox{ for } t \in [0,T], \ x \in \R^3.
\end{equation}
Integrating \eqref{eq:defWeakSolution} from 0 to $t$ and using the fact that $u(t= 0, \cdot ) = 0$, we obtain that,  for all $v \in \VV$ and almost every $t \in [0,T]$,
\begin{equation} 
\ll M u(t, \cdot), v \rr_\HH    +  \int_0^t \ll (\Lambda \star u)(s, \cdot),v \rr_\HH \,  ds = \ll  U(t, \cdot), \mathbb A v \rr_\HH.
\end{equation}
Using the fact that 
\begin{equation}
\partial_t U (t, \cdot ) = u(t, \cdot),\quad \mbox{ for a.e. } t \in (0, T), 
\end{equation}
we derive that, for all $v \in \HH$,
\begin{equation}
 \ll M \partial_t U(t, \cdot), v \rr_\HH + \int_0^t \ll (\Lambda \star u)(s, \cdot) \,  ds,v \rr_\HH = \ll  U(t, \cdot ), \mathbb A v \rr_\HH, \quad \mbox{ in } (0, T). 
\end{equation}
We claim that
\begin{equation}\label{eq:uniqueness2}
\int_0^t (\Lambda \star u)(s, \cdot) \,  ds = (\Lambda \star U)(t, \cdot), \quad \text{for almost every $t \in (0,T)$}.
\end{equation}
Indeed, by Fubini's theorem,  one gets, for almost every $t \in (0,T)$,
\begin{equation}\label{thm1-p1}
\begin{aligned}
\int_0^t (\Lambda \star u)(s, \cdot) \,  ds & = 
\int_0^t \left[ \int_0^s  \Lambda(\tau, \cdot) u(s-\tau, \cdot)\,  d\tau \right]  d s = \int_0^t \Lambda(\tau, \cdot) \left[ \int_\tau^t u(s-\tau, \cdot )\,  ds \right]  d\tau \\ 
& = \int_0^t \Lambda(\tau, \cdot) \left[ \int_0^{t-\tau} u(\tau^\prime, \cdot )\,  d\tau^\prime \right]  d\tau 
 = (\Lambda \star U)(t, \cdot).
\end{aligned}
\end{equation}

From \eqref{eq:problem}, we derive that 
\begin{equation}\label{eq-U}
 M(x) \partial_t U(t,x) + (\Lambda \star U)(t,x) =  \mathbb A U(t,x),\quad  \mbox{ for } t \in (0, T), \ x \in \R^3
\end{equation}
and hence $U \in L^1(0, T; \VV)$. 
Multiplying \eqref{eq-U} by $U(t, \cdot)$, integrating with respect to $x$,  and using Fubini's theorem as well as the fact that $\ll  v,  \mathbb Av \rr_\HH = 0$ for all $v \in \VV$, we obtain
\begin{equation}
\frac{1}{2}  \frac{ d}{ dt} \ll M U(t, \cdot ), U(t, \cdot) \rr_\HH   +   \ll  (\Lambda \star U)(t, \cdot) , U(t, \cdot) \rr_\HH  =0, \quad \text{for almost every $t \in [0,T]$}.
\end{equation}
Integrating this equation from 0 to $t$ gives
\begin{equation}
\frac{1}{2}\ll M U(t, \cdot), U(t, \cdot) \rr_\HH   +  \int_0^t  \ll  (\Lambda \star U)(s, \cdot) , U(s, \cdot) \rr_\HH  ds =0, \quad \text{for almost every $t \in [0,T]$}.
\end{equation}
We derive from  \eqref{eq:weakPassivity} that  $\| U(t) \|_\HH^2 \le 0$ for almost every $t \in [0,T]$. It follows that 
\begin{equation}
U(t, \cdot)= 0,\qquad \text{for almost every $t \in [0,T]$.}
\end{equation}
This in turn implies that $u = 0$. The proof is complete. \qed

\medskip 
In the proof of Theorem~\ref{thm:wellPosedness}, we use the following Gr\"onwall's inequality:

\begin{lemma} \label{lem-GW}
Let $T>0$, $\tau \in (0,1)$, $\alpha,\beta \ge 0$ and let $\xi$ and $\phi$ be  two non-negative, measurable functions  defined in $(0, T)$ such that
\begin{equation}
\xi(t) \le \alpha + \beta \int_0^t \phi(s) \xi (s)^\tau\,  ds,\quad \text{ for almost every $t \in (0,T)$}. 
\end{equation}
We have
\begin{equation}\label{eq:gronwall}
\xi(t) \le \left( \alpha^{1-\tau} + (1-\tau) \beta  \int_0^t \phi(s)\,  ds \right)^{1/(1-\tau)}, \quad \text{ for almost every $t \in (0,T)$}.
\end{equation}
\end{lemma}

\begin{proof} The proof of this result is standard. Set 
\begin{equation}
G(t):= \alpha + \beta \int_0^t \phi(s) \xi (s)^\tau\,  ds, \quad \mbox{ for } t \in (0,T).
\end{equation}
Then  $ G^\prime(t) = \beta \phi(t) \xi(t)^\tau \le \beta \phi(t)G(t)^\tau$ for $t \in (0,T)$ and consequently
\begin{equation}
G(t)^{-\tau}G^\prime(t) \le \beta \phi(t),  \quad \mbox{ for } t \in (0,T).
\end{equation}
Integrating this with respect to $t$ and using the fact  $G(t) \ge \xi(t)$ for  $t \in (0,T)$ yield the conclusion. 
\end{proof}

\begin{remark} \rm In \cite{NguyenVogelius}, the authors used Lorentz's model to study approximate cloaking via a change of variables for the acoustic waves in the time domain.  Waves equations which are non-local in time  also appeared in a very different context in \cite{MinhLinh}, the one of  generalized impedance boundary conditions for conducting obstacles. The proof has some roots in these works. 
\end{remark}

\begin{remark} \label{rem-Unique} \rm Assume that \eqref{epsmu}, \eqref{assumption-lambda}, \eqref{eq:causality} and
\eqref{eq:passivity} hold. Let $u \in L^\infty(-\infty, + \infty; \HH)$ be a weak solution of 
\begin{equation}\label{eq-G}
M(x) \partial_t u(t,x) + (\Lambda * u)(t,x) = \mathbb A u(t,x) + f(t,x),\quad \mbox{ for } t\in R, \ x \in \R^3. 
\end{equation}
Note that the time convolution $*$ is considered here, and not the operator $\star$ defined by \eqref{eq:truncatedConvolutionProduct}. Assume that $f(t, \cdot) = 0$ for $t <  t_1$ and in addition that $ u \in L^1(-\infty, t_1; \VV)$ and 
\begin{equation}\label{minus-infinity}
\liminf_{t \to - \infty} \| u(t, ) \|_{\HH} = 0.
\end{equation}
Then $u(t, \cdot) = 0$ for $t < t_1$.  The definition of weak solutions for \eqref{eq-G} is similar to the one given in Definition \ref{defi:weakSolution}: $u$ is required to satisfy the following equation,  in the distributional sense, 
\begin{equation}\label{WS-G}
\frac{ d}{ dt} \ll M u(t, \cdot ),v \rr_\HH + \ll (\Lambda * u)(t, \cdot),v\rr_\HH   =  \ll u(t, \cdot), \mathbb A v \rr_\HH  + \ll f(t, \cdot),v \rr_\HH,\quad  \mbox{ in } (-\infty, \infty),
\end{equation}
for all $v \in \VV$. Indeed, we have 
\begin{equation}
\frac{1}{2} \frac{ d}{ dt} \ll M u(t, \cdot ),u(t, \cdot) \rr_\HH + \ll (\Lambda * u)(t, \cdot),u(t, \cdot)\rr_\HH   =  0,\quad  \mbox{ in } (-\infty, t_1). 
\end{equation}
This implies, by \eqref{minus-infinity},  
\begin{equation}
\frac{1}{2} \ll M u(t, \cdot ),u(t, \cdot) \rr_\HH + \int_{-\infty}^t \ll (\Lambda * u)(t, \cdot),u(t, \cdot)\rr_\HH   =  0,\quad  \mbox{ in } (-\infty, t_1). 
\end{equation}
Similar to \eqref{eq:weakPassivity}, we obtain, for $t < t_1$,  
\begin{equation}
 \int_{-\infty}^t  \ll (\Lambda * u)(s, \cdot),  u(s, \cdot) \rr_\HH  d s \ge 0. 
\end{equation}
Therefore,  $u(t, \cdot) = 0$ for $t < t_1$.  
\end{remark}

\subsection{Proof of Theorem \ref{thm:finitePropagationSpeed}}\label{sec:fps}

In the case where $u$ is {\it regular enough}, the argument is quite standard using the two observations \eqref{eq:weakPassivity} and  \eqref{integration-Auu}. To overcome the lack of the regularity of $u$,  we implement the strategy used in the proof of the uniqueness part of Theorem~\ref{thm:wellPosedness}. For simplicity of notations, we assume that $a = 0$ and we denote $c_{a, R}$ by $c$ in this proof.

Set
\begin{equation}
U(t, x) := \int_0^t u(s, x)\, ds, \quad \mbox{ for } t \in [0,T), \ x \in \R^3.
\end{equation}
Integrating \eqref{eq:defWeakSolution} from 0 to $t$ and using the fact that $u(t= 0, \cdot ) = u_0$, we obtain that, for all $v \in \VV$ and for almost every $t \in (0,T)$,
\begin{equation} 
\ll M u(t, \cdot), v \rr_\HH - \ll M u_0, v \rr_\HH +  \int_0^t \ll (\Lambda \star u)(s, \cdot),v \rr_\HH \,  ds = \ll  U(t, \cdot), \mathbb A v \rr_\HH + \ll F(t), v \rr_\HH, 
\end{equation}
where 
\begin{equation} 
F(t, \cdot) := \int_0^t f(s, \cdot) \, ds, \quad \mbox{ for } t \in [0,T).
\end{equation}
As in \eqref{thm1-p1}, we have
\begin{equation}
\int_0^t (\Lambda \star u)(s, \cdot) \,  ds = (\Lambda \star U)(t, \cdot), \quad \mbox{ for almost every } t \in [0,T).
\end{equation}
Since 
\begin{equation}\label{thm3-p1}
\partial_t U (t, \cdot ) = u(t, \cdot),\quad \mbox{ for almost every } t \in (0, T), 
\end{equation}
we derive,  for all $v \in \HH$, that in $(0,T)$
\begin{equation}
 \ll M \partial_t U(t, \cdot), v \rr_\HH +  \ll (\Lambda \star U)(s, \cdot) \,  ds,v \rr_\HH = \ll  U(t, \cdot ), \mathbb A v \rr_\HH + \ll F(t, \cdot),\phi_k \rr_\HH + \ll M u_0, v \rr_\HH.
\end{equation}
It follows that 
\begin{equation}\label{eq-U2}
M(x) \partial_t U (t, x) + (\Lambda \star U) (t, x)=  \mathbb A U (t, x) +  F(t, x) +  M u_0(x),\quad  \mbox{ for } t \in (0, T), \ x \in \R^3.  
\end{equation}
From \eqref{thm3-p1}, we obtain 
\begin{equation}\label{thm3-p2}
 U \in L^1(0, T; \VV). 
\end{equation}
We claim
\begin{equation}\label{claim-thm2}
U(t, \cdot) = 0,\quad   \mbox{ in } B_{R - ct} \mbox{ and for } t \in (0, R/ c). 
\end{equation}
Since $u_0 = 0$ in $B_R$, it is clear that the conclusion follows from claim~\eqref{claim-thm2} and the definition of $U$. 

It remains to prove \eqref{claim-thm2}. 
Multiplying the equation of $U$ \eqref{eq-U2} by $U(t, x)$, integrating with respect to $x$ in $B_{R - c t}$,  and using the facts 
$u_0 = 0$ in $B_{R - ct}$ and $F(t, \cdot) = 0$ in $B_{R - ct}$ for almost every $t \in (0, R/ c)$, 
we have, for almost every $t \in (0, R/ c)$,
\begin{multline}\label{eq:tmpFPS0}
\int_{B_{R- c t}}  M \partial_t U(t, x ) \cdot U(t, x) \, d x 
+ \int_{B_{R- c t}}  (\Lambda \star U)(t, x) \cdot U(t, x) \, dx \\
 = \int_{B_{R- c t}}  \mathbb A U(t, x) \cdot U(t, x) \, dx. 
\end{multline}
The divergence theorem gives,  for almost every $t \in (0, R/ c)$,
\begin{multline}
\frac{1}{2} \frac{ d}{ dt} \int_{B_{R- c t}}  M U(t, x) \cdot U(t, x) \, dx \\ = \int_{B_{R- c t}}  M  \partial_t U(t, x) \cdot U(t, x) \, dx  - \frac{c}{2} \int_{\partial B_{R- c t}}  M U(t, x) \cdot U(t, x) \, dx.  
\end{multline}
It follows from \eqref{eq:tmpFPS0} that, for almost every $t \in (0, R/ c)$,
\begin{multline}
\frac{1}{2} \frac{ d}{ dt} \int_{B_{R- c t}}  M U(t, x) \cdot U(t, x) \, dx + \int_{B_{R- c t}}  (\Lambda \star U)(t, x) \cdot u(t, x) \, dx  \\
 = - \frac{c}{2} \int_{\partial B_{R- c t}}  M U(t, x) \cdot U(t, x) \, dx   + \int_{B_{R- c t}}  \mathbb A U(t, x) \cdot U(t, x) \, dx. 
\end{multline}
Integrating this identity from $0$ to $t$ with $ t  \in (0,R/ c)$ and using the fact that $U(0, \cdot) = 0$, we obtain, for almost every $t \in (0, R/ c)$,
\begin{multline}\label{to0}
\frac{1}{2}\int_{B_{R- c t}}   M U(t, x) \cdot U(t, x) \, dx +  \int_0^t \int_{B_{R- c s}}  (\Lambda \star U)(s, x) \cdot U(s, x) \, dx \, ds  \\
 =  - \frac{c}{2}  \int_0^t \int_{\partial B_{R- c s}}  M U(s, x) \cdot U(s, x) \, dx \, ds    +  \int_0^t \int_{B_{R- c s}}  \mathbb A U(s, x) \cdot U(s, x) \, dx \, ds. 
\end{multline}
Similar to \eqref{eq:weakPassivity}, we have 
\begin{equation}\label{to1}
\int_0^t \int_{B_{R- c s}}  (\Lambda \star U)(s, x) \cdot U(s, x) \, dx \, ds \ge 0, \quad \mbox{ for almost every } t \in (0, R/ c).
\end{equation}
Combining \eqref{to0} and \eqref{to1} yields, for almost every $t \in (0, R/ c)$,
\begin{multline}\label{to2}
\frac{1}{2}\int_{B_{R- c t}}   M U(t, x) \cdot U(t, x) \, dx \le  - \frac{c}{2}  \int_0^t \int_{\partial B_{R- c s}}  M U(s, x) \cdot U(s, x) \, dx \, ds   \\ 
 + \int_0^t \int_{B_{R- c s}}  \mathbb A U(s, x) \cdot U(s, x) \, dx \, ds .  
\end{multline}
We claim that, for $0< s < R/ c$, one has
\begin{equation}\label{eq:claim}
 - \frac{c}{2} \int_{\partial B_{R- c s}}  M U(s, x) \cdot U(s, x) \, dx    + \int_{B_{R- c s}}   \mathbb A U(s, x) \cdot U(s, x) \, dx \le 0. 
\end{equation}
Indeed, for $U = (E,H)^ T$,  one has
\begin{align*}
\int_{B_{R- c s}}  \mathbb A U(s, x) \cdot U(s, x) \, dx 
& = \int_{B(0,R-cs)} \left[ \curl H(s,x) \cdot E(s,x)- \curl E(s,x) \cdot H(s,x) \right]  dx \\
&  = - \int_{\partial B_{R-cs}} \big(H(s,x) \times e_r \big) \cdot E(s,x)\,  d x \le   \int_{\partial B_{R-cs}} |H| |E| \,  d x 
\end{align*}
and 
\begin{equation}
M U(s, x) \cdot U(s, x)  = \eps_{rel} E \cdot  E + \mu_{rel} H \cdot H  \ge 2 |\eps_{rel}^{1/2} E| |\mu_{rel}^{1/2} H|. 
\end{equation}
Assertion~\eqref{eq:claim} now follows from the definition \eqref{eq:defSpeed} of $c = c_{a, R}$.

We derive from \eqref{to2} and \eqref{eq:claim} that 
\begin{equation}
\frac{1}{2} \int_{B_{R- c t}}   M U(t, x) \cdot U(t, x) \, dx \le 0,  \quad 
\mbox{ for almost every } t \in (0, R/ c),
\end{equation}
and claim~\eqref{claim-thm2} follows from the ellipticity of $M$. \qed

\subsection{Proof of Theorem~\ref{thm-reg}}\label{sect-reg} In this proof, we use the notations from the one of Theorem~\ref{thm:wellPosedness}. For $n \in \N^*$, set 
\begin{equation*}
v_n(t, x): = \partial_t u_n(t, x), \quad  \mbox{ for } t \in [0, T), \ x \in \R^3. 
\end{equation*}
We recall that $u_n$ is the approximate solution constructed by the Galerkin approach in the proof of Theorem \ref{thm:wellPosedness}. It follows from \eqref{eq:defun} that 
\begin{equation}\label{form-vn}
v_n(t, x) =  \sum_{k=1}^n d_{n,k}'(t) \phi_k(x), \quad  \mbox{ for } t \in [0, T), \ x \in \R^3. 
\end{equation}
Differentiating \eqref{eq:approximateWeakSolution} with respect to $t$, we have 
\begin{equation}\label{vn-11}
 \frac{ d}{ dt} \ll M v_n(t, \cdot),\phi_k \rr_\HH + \ll (\Lambda \star v_n)(t, \cdot),\phi_k \rr_\HH = \ll v_n(t, \cdot), \mathbb A\, \phi_k \rr_\HH +\ll g_n(t, \cdot),\phi_k \rr_\HH, , \quad  \mbox{ in } (0,T),
\end{equation}
where 
\begin{equation}\label{def-gn}
g_n(t, x) := \partial_t f(t, x) - \Lambda(t, x) u_{0, n} (x), \quad  \mbox{ for } t \in (0, T), \ x \in \R^3. 
\end{equation}
We have 
\begin{equation}\label{form-vn-1}
 \ll M \partial_t u_n(0, \cdot),\phi_k \rr_\HH = \ll u_n(0, \cdot), \mathbb A\, \phi_k \rr_\HH +\ll f(0, \cdot),\phi_k \rr_\HH, \quad  \mbox{ for }k \in \{1, \cdots, n\}.
\end{equation}
It follows from \eqref{form-vn-1} that $
M^{1/2} \partial_t u_n (0, \cdot)$ is the projection of  $M^{-1/2} \big(\mathbb A u_0  + f(0, \cdot) \big)$
into the space spanned by  $\{ M^{1/2}\phi_1, \cdots, M^{1/2}\phi_n \}$ in  $\HH$. 
This implies 
\begin{equation}
\| v_n(0, \cdot) \|_{\HH} = \| \partial_t u_n(0, \cdot) \|_{\HH} \le C \Big( \| u_0 \|_\VV + \| f(0, \cdot) \|_\HH \Big). 
\end{equation}
As in \eqref{eq:approximateAPrioriEstimate}, we derive from \eqref{vn-11} that 
\begin{equation}
\| v_n(t, \cdot)\|_\HH^2  \le  C \left(  \| u_0 \|_\VV + \| f(0, \cdot) \|_\HH + \int_0^t \|\partial_s f (s, \cdot) \|_{\HH} + \| \Lambda(s, \cdot)\|_{L^\infty} \| u_n(s. \cdot) \|_{\HH} \, ds \right)^2 \quad \mbox{ in } (0, T). 
\end{equation}
This in turn yields
\begin{equation}
\| v(t, \cdot)\|_\HH^2  \le  C \left(  \| u_0 \|_\VV + \| f(0, \cdot) \|_\HH + \int_0^t \|\partial_s f (s, \cdot) \|_{\HH} + \| \Lambda(s, \cdot)\|_{L^\infty} \| u(s. \cdot) \|_{\HH}  \, ds \right)^2 \quad \mbox{ in } (0, T), 
\end{equation}
and the conclusion follows. \qed

\section{Numerical results}\label{sect:num}

We now perform some numerical simulations. In this section we focus on the Drude's model without absorption described at the end of Section \ref{sect:model}. More precisely we consider $\eps_{rel} = \mu_{rel} = 1$,  $\widehat{\lambda_{em}} = \widehat{\lambda_{me}} = 0$,
\begin{equation} 
\widehat{\lambda_{ee}}(\omega,x) =   \frac{\oomega_e^2(x)}{-i \omega} \quad \text{and} \quad 
\widehat{\lambda_{mm}}(\omega,x) = \frac{\oomega_m^2(x)}{-i \omega}, \qquad \mbox{ for } \omega \in \R,\ x \in \R^3,
\end{equation}
or equivalently
\begin{equation} 
\lambda_{ee}(t,x) = \oomega_e^2(x)\theta(t) \quad \text{and} \quad 
\lambda_{mm}(t,x) = \oomega_m^2(x)\theta(t),  \qquad \mbox{ for } t \in \R,\ x \in \R^3,
\end{equation}
where $\oomega_e$ and $\oomega_m$ are two functions defined later. 

In this context, the problem \eqref{eq:problem0} rewrites
\begin{equation}\label{eq:generalMaxwellEquationsDrude}
\left\{ \begin{aligned}
&  \partial_t E(t, x) + \oomega_e^2(x)\int_0^t E(s,x)\, ds  =  \curl H(t, x) +f_e(t,x),  \\[6pt]
&   \partial_t H(t, x)+  \oomega_m^2(x)\int_0^t H(s,x)\, ds = -\curl E(t, x) +f_m(t,x), \\[6pt]
& E(t=0,\cdot) = E_0, \qquad H(t=0,\cdot) = H_0. 
\end{aligned}\right. \qquad \mbox{ for } t > 0,\ x \in \R^3. 
\end{equation}



Define
\begin{equation}\label{eq:defAuxillaryFields}
J(t,x) := \int_0^t E(s,x)\, ds \quad \text{and} \quad
K(t,x) := \int_0^t H(s,x) \, ds, \qquad \mbox{ for } t \ge 0,\ x \in \R^3.
\end{equation}
It is clear that $J(t=0,\cdot) = K(t=0,\cdot) = 0$ and that
\begin{equation}
\partial_t J(t,x) = E(t,x) \quad \text{and} \quad
\partial_t K(t,x) = H(t,x), \qquad \mbox{ for } t \ge 0,\ x \in \R^3.
\end{equation}
From \eqref{eq:generalMaxwellEquationsDrude}, one obtains the following local in time problem which  is the advantage of the special structure of Drude's model:
\begin{equation}\label{eq:generalMaxwellEquationsDrude2}
\left\{ \begin{aligned}
&  \partial_t E(t, x) + \oomega_e^2(x)J(t,x)  =  \curl H(t, x) + f_e(t, x), \\
&   \partial_t H(t, x)+  \oomega_m^2(x)K(t,x) = -\curl E(t, x) + f_m (t, x), \\
& \partial_t J(t,x) = E(t,x), \\
& \partial_t K(t,x) = H(t,x), \\
& E(t=0,\cdot) = E_0, \qquad J(t=0,\cdot) = 0, \\
& H(t=0,\cdot) = H_0, \qquad K(t=0,\cdot) = 0. 
\end{aligned}\right. \qquad \mbox{ for }t > 0,\ x \in \R^3. 
\end{equation}
%

We are interested in simulations on  \eqref{eq:generalMaxwellEquationsDrude2} in the 2d setting for simplicity. We thus  consider the case in which $(E_0,H_0)$, $(f_e,f_m)$,  and $(\oomega_e,\oomega_m)$ do not depend on the third variable $x_3$ in space (here $x = (x^\prime,x_3) \in \R^3$ with $x^\prime =(x_1,x_2) \in \R^2$). One can show that the four fields $E$, $H$, $J$ and $K$ are also independent of $x_3$ and that  one has the two decoupled systems respectively called \emph{transverse-electric} and \emph{transverse-magnetic} modes. Here we focus on the transverse-electric modes, which are given as follows, for $t > 0$ and $x^\prime =(x_1,x_2) \in \R^2$:
\begin{equation}\label{eq:TM}
\left\{ \begin{aligned}
&   \partial_t E_3(t,x^\prime)+  \oomega_e^2(x_1,x_2)J_3(t,x^\prime) = \partial_{x_1} H_2(t,x^\prime) - \partial_{x_2} H_1(t,x^\prime) +f_{e,3}(t,x^\prime) , \\[6pt]
&  \partial_t H_1(t,x^\prime) + \oomega_m^2(x^\prime)K_1(t,x^\prime)  =  -\partial_{x_2} E_3(t,x^\prime) + f_{m,1}(t,x^\prime), \\[6pt]
&  \partial_t H_2(t,x^\prime) + \oomega_m^2(x^\prime)K_2(t,x^\prime)  =  \partial_{x_1} E_3(t,x^\prime)+ f_{m,2}(t,x^\prime), \\[6pt]
& \partial_t J_3(t,x^\prime) = E_3(t,x^\prime),  \; \partial_t K_1(t,x^\prime) = H_1(t,x^\prime),\;   \partial_t K_2(t,x^\prime) = H_2(t,x^\prime), \\[6pt]
& E_3(t=0,\cdot) = E_{0,3}, \;  H_1(t=0,\cdot) = H_{0,1}, \;  H_2(t=0,\cdot) = H_{0,2},\\[6pt]
& J_3(t=0,\cdot) = 0,  \; K_1(t=0,\cdot) = 0,  \;  K_2(t=0,\cdot) = 0.
\end{aligned}\right. 
\end{equation}

The setting for the simulation is the following.  The medium consists of  a bounded rectangular obstacle filled with a Drude's material  with positive constant  $\omega_e$ and $\omega_m$  that is surrounded by vacuum, i.e., $(\oomega_e, \oomega_m) = (\omega_e, \omega_m)$ inside the rectangle and $(0, 0)$ otherwise
(see Figure \ref{fig:geometry}).  We impose zero initial conditions for the electric and the magnetic fields: 
\begin{equation}
E_3(t=0,\cdot) =  H_1(t=0,\cdot) = H_2(t=0,\cdot) = 0,
\end{equation}
and zero magnetic sources:
\begin{equation}
f_{m,1} = f_{m,2} = 0.
\end{equation}
We choose
\begin{equation}\label{eq:source}
f_{e,3}(t,x_1,x_2) = \sin(\omega_* t) g(x_1,x_2), \quad \mbox{ for } t >0,\ (x_1,x_2) \in \R^2,
\end{equation}
where $g$ is a Gaussian given by 
\begin{equation}
g(x_1,x_2) = e^{-25(x_1+10)^2-25x_2^2}, \quad \mbox{ for }  (x_1,x_2) \in \R^2.
\end{equation}
By selecting appropriately $\omega_e$, $\omega_m$ and $\omega_*$,  the obstacle can have a negative permittivity, a negative permeability or even both.

Concerning the numerical methods, we use classical PMLs to artificially bound the computational domain and for the numerical scheme we use $P^1$-$P^0$ mixed finite elements (with mass lumping for efficiency) for the space discretization and centred finite difference approximations on staggered grids for the time discretization. The computations were done with FreeFem++ \cite{hecht}. We refer to \cite{vinoles} for more details about these numerical methods.


We perform  three numerical experiments.
\begin{itemize}
\item  In the first one, we take $\omega_* = 5$, $\omega_e = 4$ and $\omega_m = 2$. With this choice we have
\begin{equation*}
\widehat{\eps}(\omega_*) \simeq  0.36 > 0 
\quad \text{and} \quad
\widehat{\mu}(\omega_*) \simeq 0.84 > 0. 
\end{equation*}
Here, the ``effective'' permittivity and permeability are both positive. Figure \ref{fig:expe1} shows some snapshots of $E_3$ at different times. One can see that there is propagation inside the obstacle, but with different speed (and consequently wavelength). This is due to dispersion.

\item In the second simulation, we take $\omega_* = 5$, $\omega_e = 6$ and $\omega_m = 2$. With this choice we have
\begin{equation*}
\widehat{\eps}(\omega_*) \simeq -0.44 < 0 
\quad \text{and} \quad
\widehat{\mu}(\omega_*)  \simeq 0.84 > 0. 
\end{equation*}
Here, the ``effective'' permittivity and permeability are of opposite signs. Figure \ref{fig:expe2} shows some snapshots of $E_3$ at different times. One can see that there is no propagation inside the obstacle: the field is exponentially decaying (after the transient wave has passed).
\item In the third simulation, we take $\omega_* = 5$, $\omega_e = 5 \sqrt 2$ and $\omega_m = 5 \sqrt 2$. With this choice we have
\begin{equation*}
\widehat{\eps}(\omega_*) \simeq -1 < 0 
\quad \text{and} \quad
\widehat{\mu}(\omega_*) \simeq  -1 < 0. 
\end{equation*}
Here, the ``effective'' permittivity and permeability are both negative. Figure \ref{fig:expe3} shows some snapshots of $E_3$ at different times. There is propagation inside the obstacle. The field focuses inside the obstacle and re-focuses symmetrically to the source outside the obstacle on the right. 
\end{itemize}

\begin{figure}[!h]
\begin{tikzpicture}
\draw (-1,-2) -- (1,-2) -- (1,-2) -- (1,2) -- (-1,2) -- cycle;
\draw (-3.5,-3.5) -- (3.5,-3.5) -- (3.5,-3.5) -- (3.5,3.5) -- (-3.5,3.5) -- cycle;
\draw[dashed] (-3.5,-3) -- (3.5,-3);
\draw[dashed] (-3.5,3) -- (3.5,3);
\draw[dashed] (3,-3.5) -- (3,3.5);
\draw[dashed] (-3,-3.5) -- (-3,3.5);
\draw (0,0) node[align=center]{obstacle \\ $\omega_e>0$ \\ $\omega_m > 0$};
\draw (0,2.5) node{vacuum $\omega_e = \omega_m = 0$};
\draw (0,3.25) node{PML};
\draw (0,-3.25) node{PML};
\draw (-3.25,0) node[align=center]{P \\ M \\ L};
\draw (3.25,0) node[align=center]{P \\ M \\ L};
\draw[<->] (-1,-2.25) -- (1,-2.25);
\draw (0,-2.25) node[below]{10};
\draw[<->] (1.25,-2) -- (1.25,2);
\draw (1.25,0) node[right]{20};
\draw[<->] (-3.5,-3.75) -- (-3,-3.75);
\draw[<->] (-3,-3.75) -- (3,-3.75);
\draw[<->] (3,-3.75) --(3.5,-3.75);
\draw (-3.25,-3.75) node[below]{3};
\draw (0,-3.75) node[below]{30};
\draw (3.25,-3.75) node[below]{3};
\draw[<->] (3.75,-3.5) -- (3.75,-3);
\draw[<->] (3.75,-3) -- (3.75,3);
\draw[<->] (3.75,3) -- (3.75,3.5);
\draw (3.75,3.25) node[right]{3};
\draw (3.75,0) node[right]{30};
\draw (3.75,-3.25) node[right]{3};
\draw (-2,0) circle (0.7);
\draw (-2,0) node[align=center]{source \\ $f_{e,3}$};
\end{tikzpicture}
\caption{Geometry of the problem \eqref{eq:TM}}
\label{fig:geometry}
\end{figure}
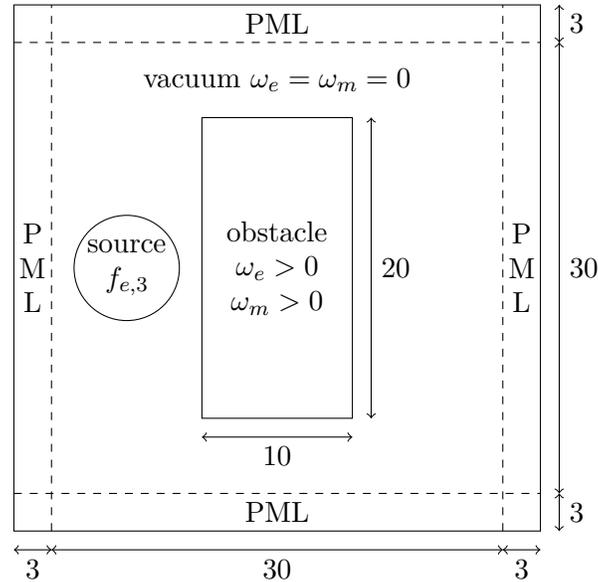

\begin{figure}[!h]
\caption{Snapshots of $E_3$ at different times for the first experiment.}

\includegraphics[width=0.49\textwidth]{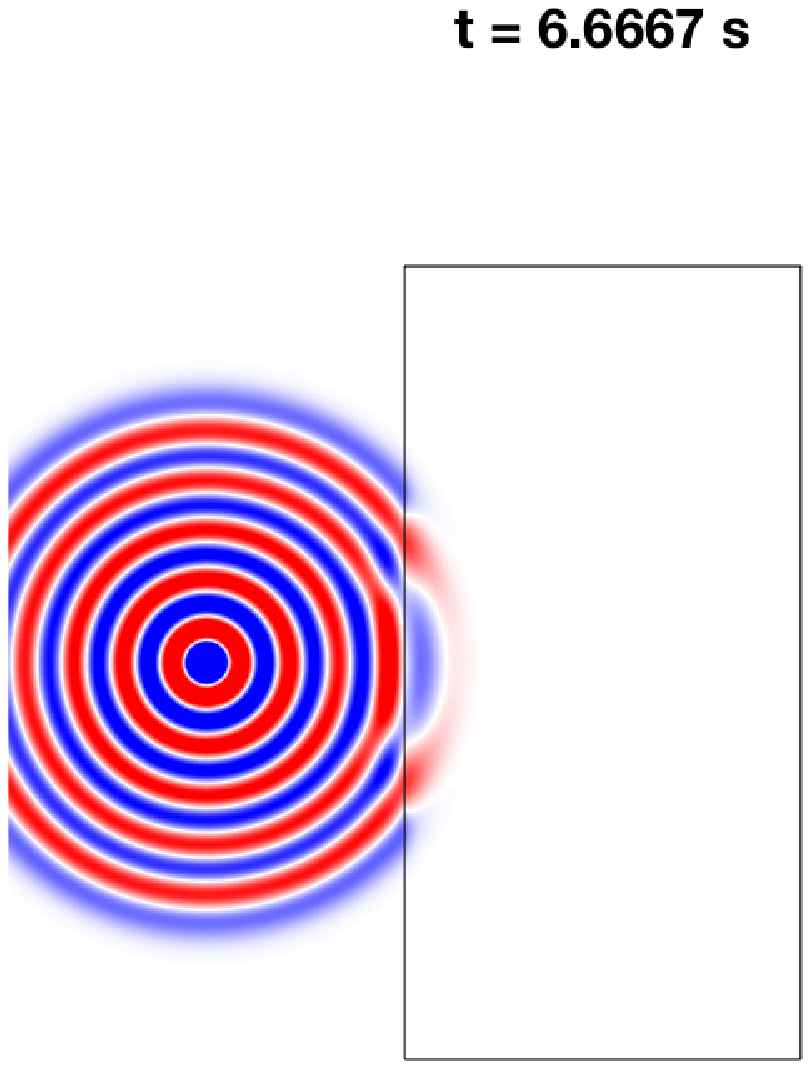}
\hfill
\includegraphics[width=0.49\textwidth]{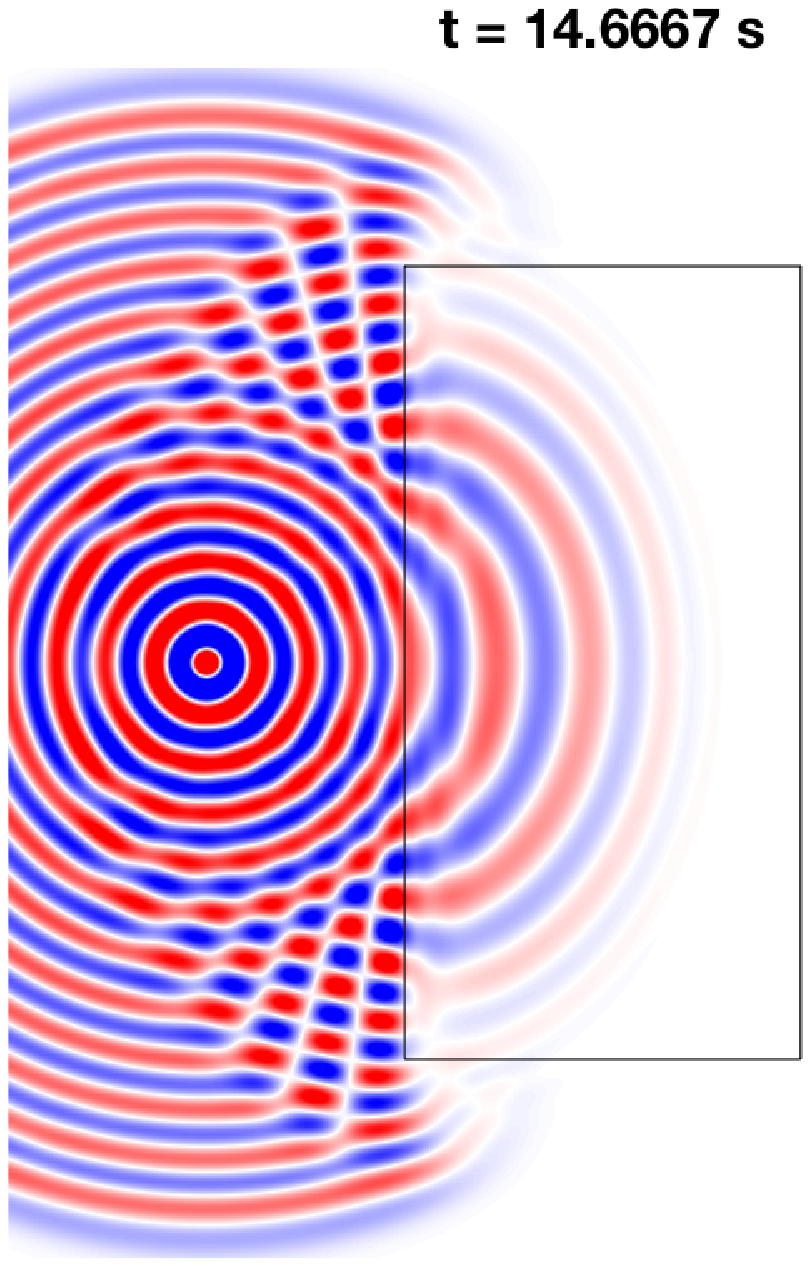}

\includegraphics[width=0.49\textwidth]{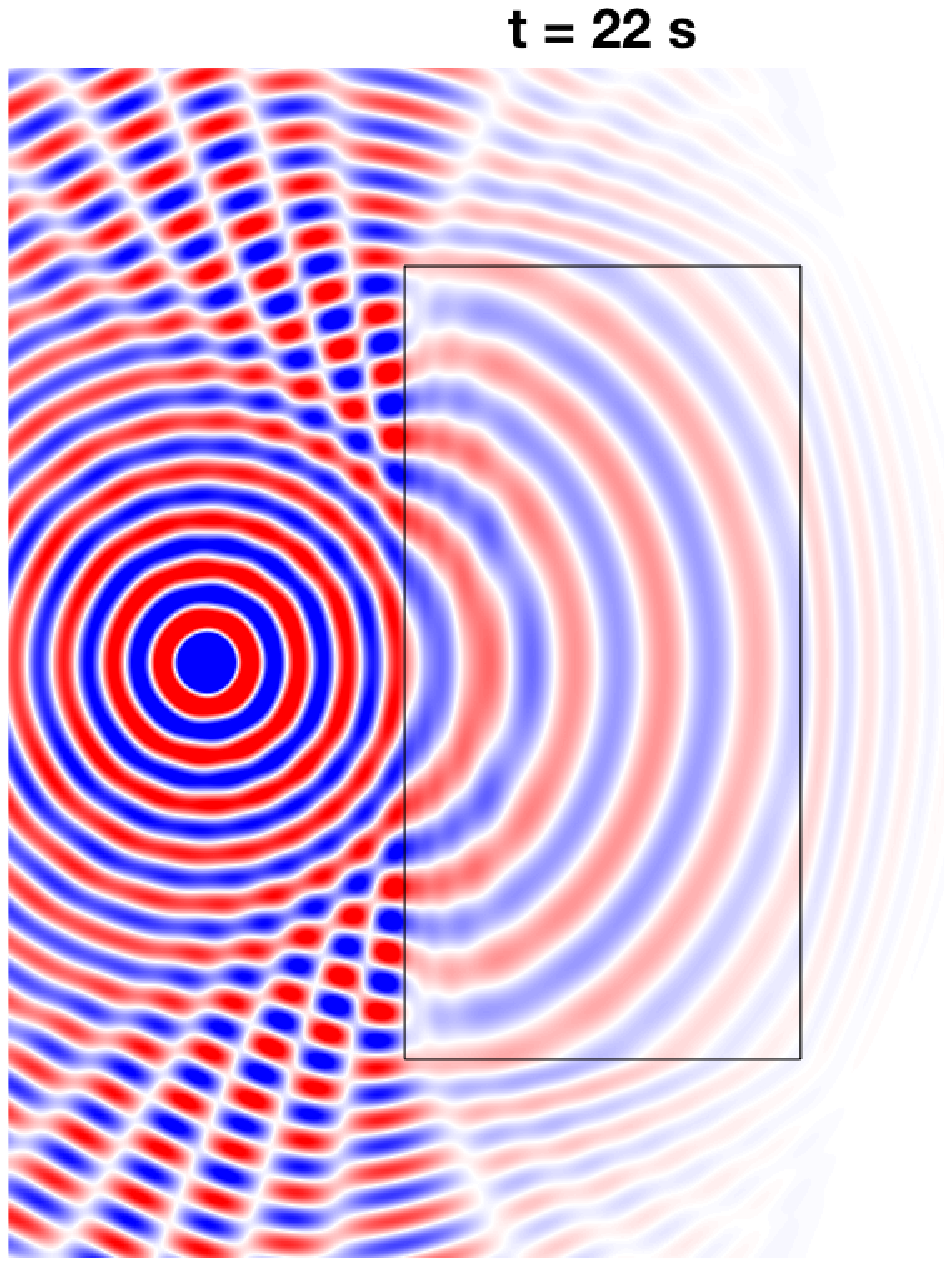}
\hfill
\includegraphics[width=0.49\textwidth]{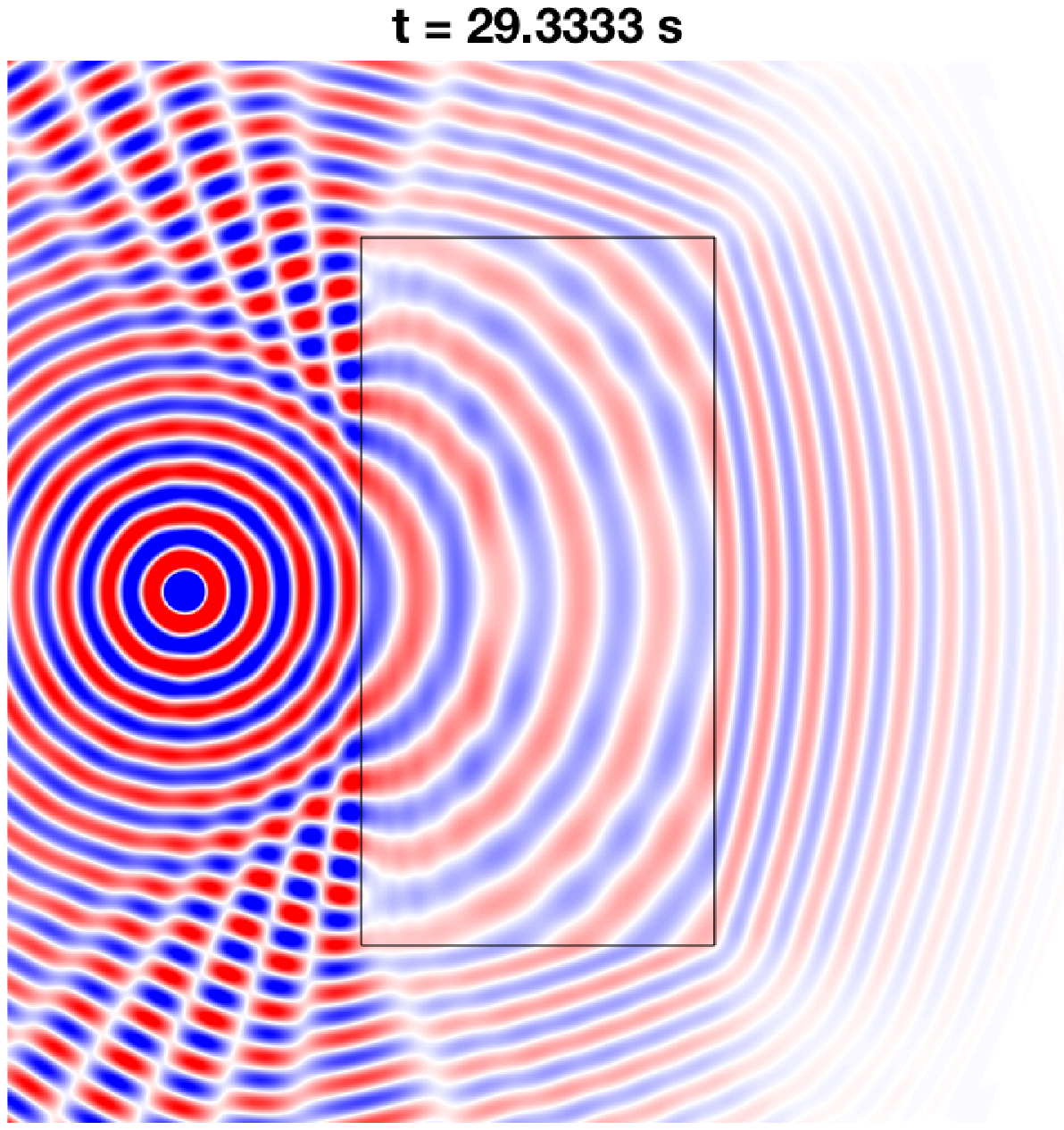}

\includegraphics[width=0.49\textwidth]{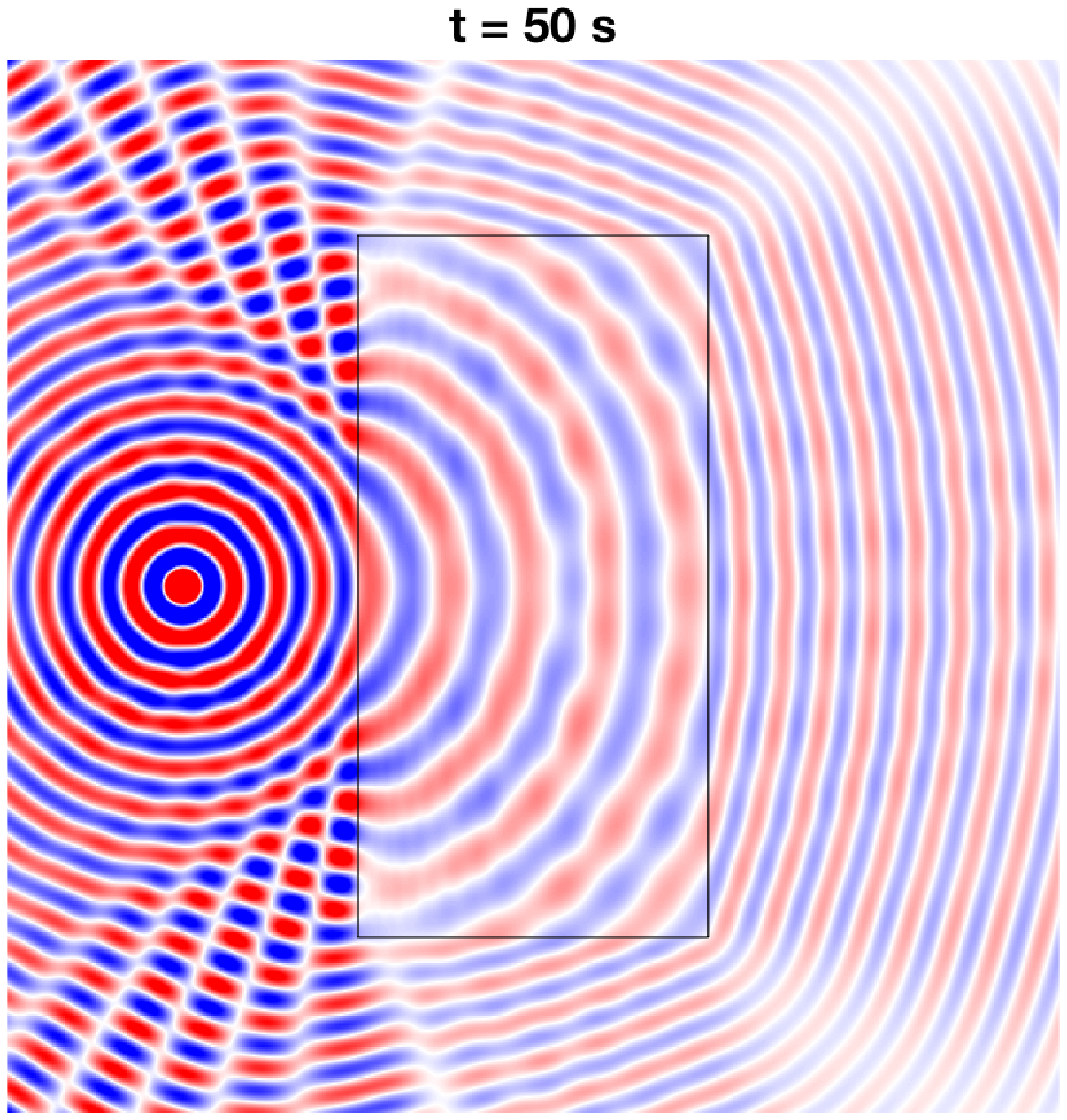}
\hfill
\includegraphics[width=0.49\textwidth]{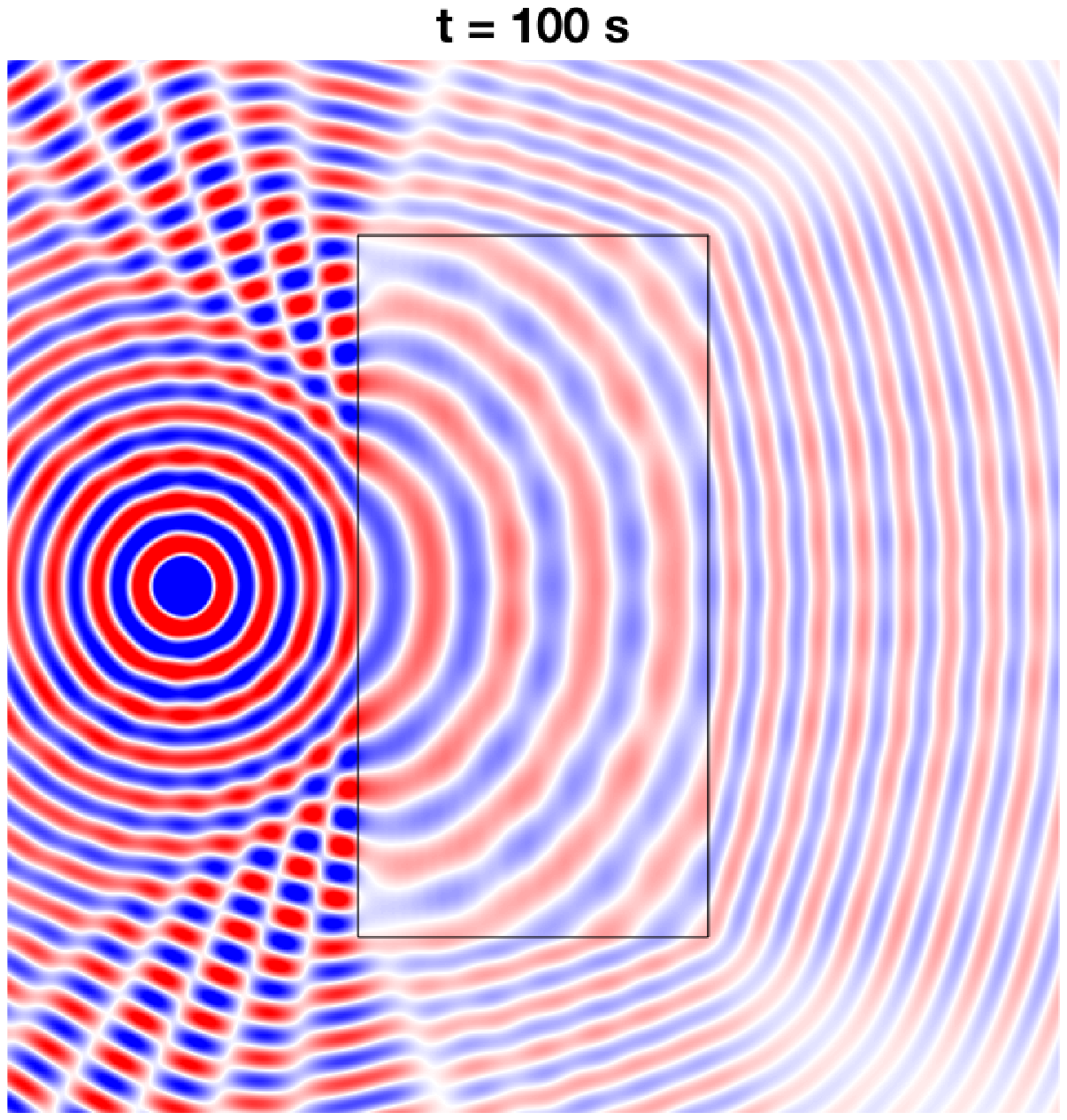}

\label{fig:expe1}
\end{figure}

\begin{figure}[!h]
\caption{Snapshots of $E_3$ at different times for the second experiment.}

\includegraphics[width=0.49\textwidth]{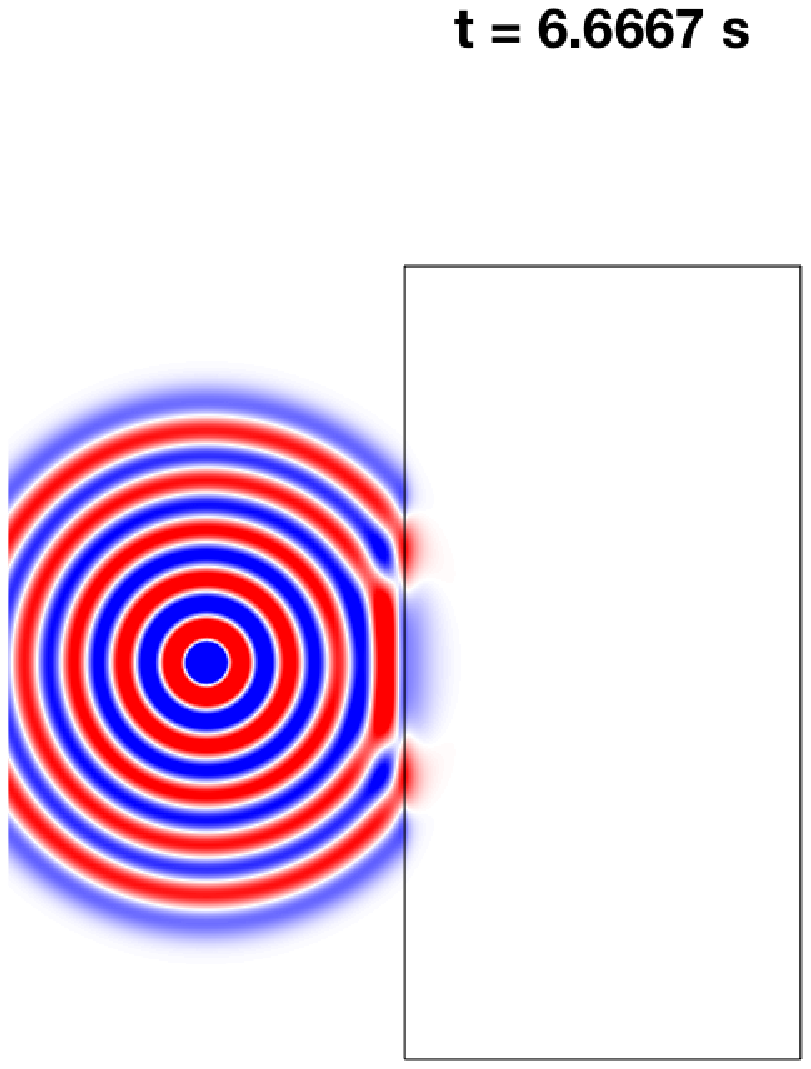}
\hfill
\includegraphics[width=0.49\textwidth]{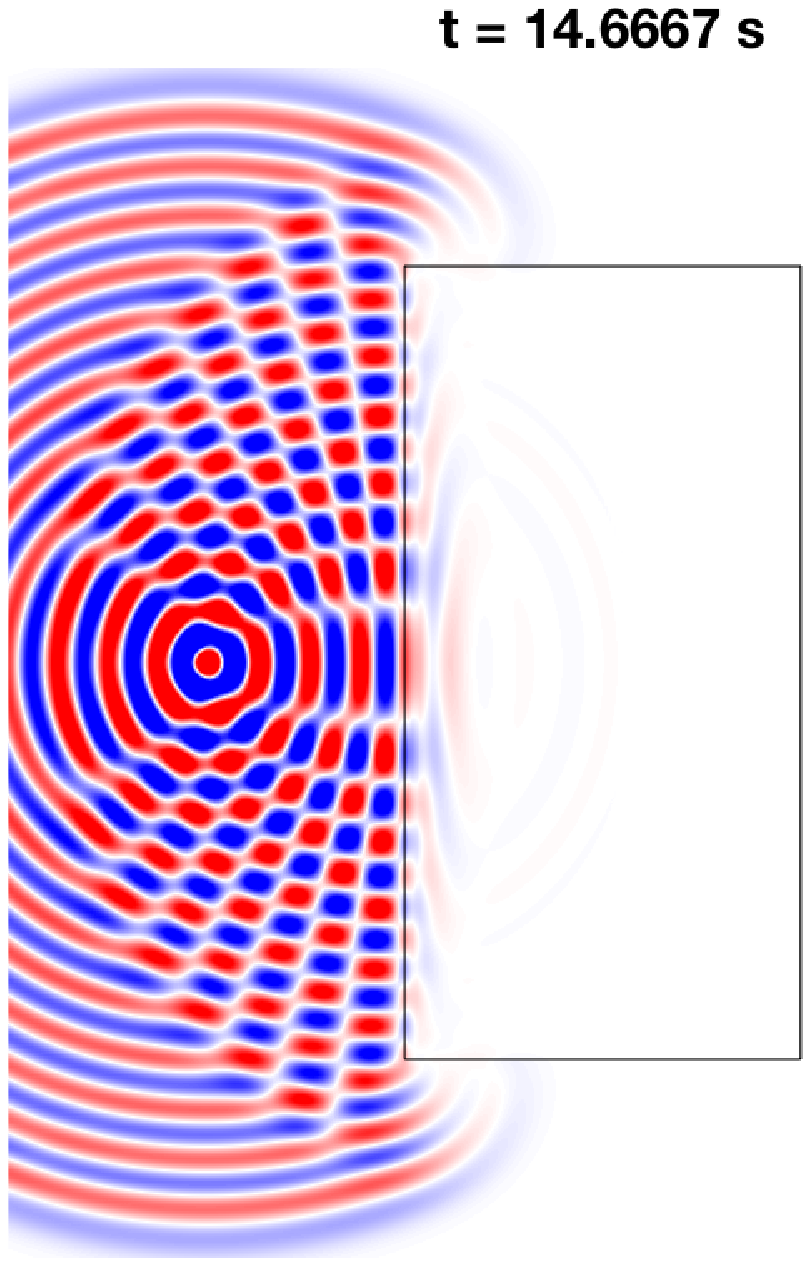}

\includegraphics[width=0.49\textwidth]{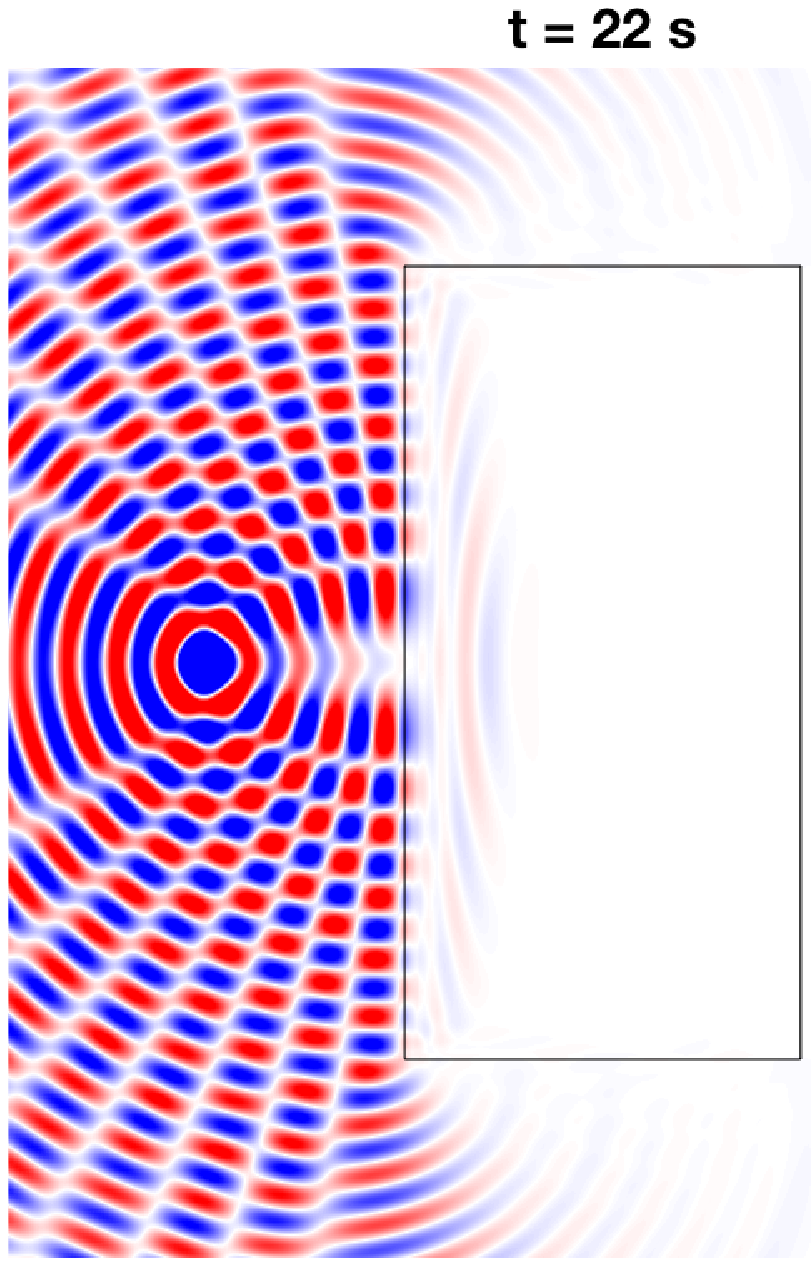}
\hfill
\includegraphics[width=0.49\textwidth]{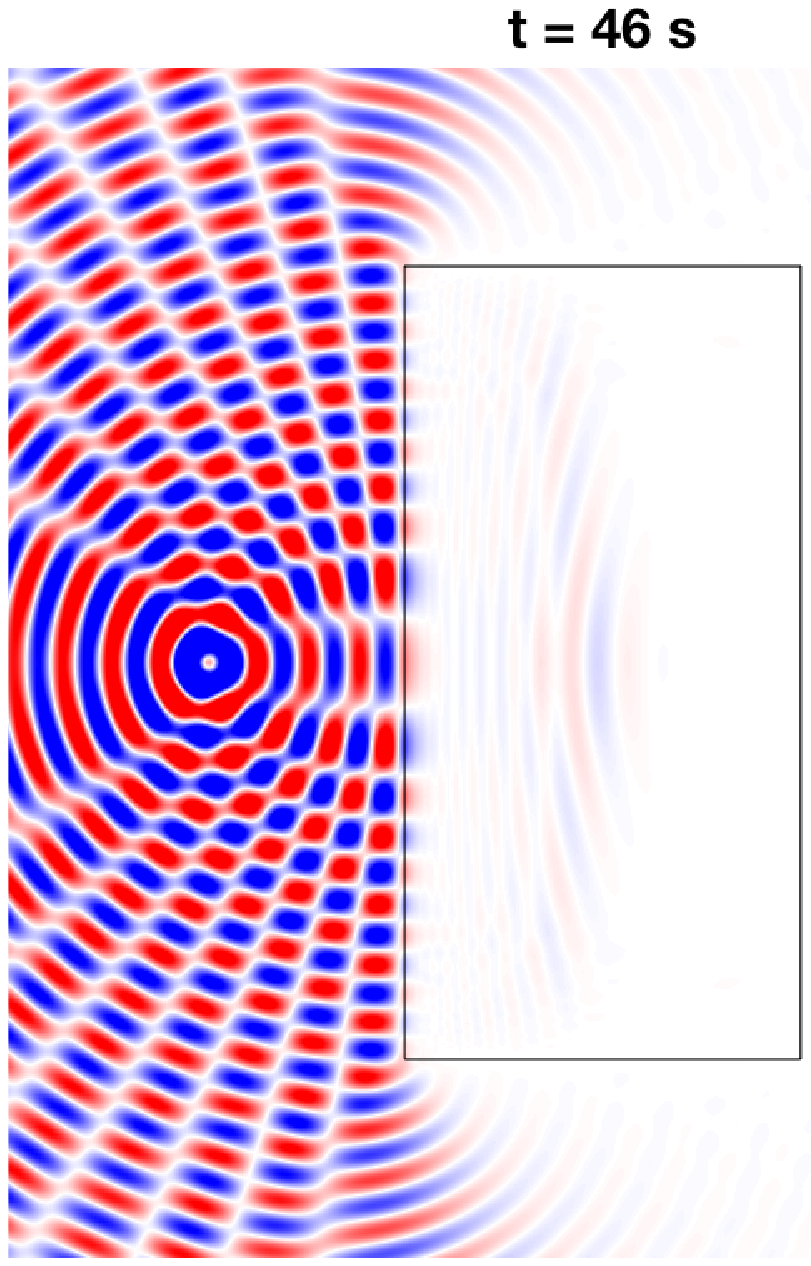}

\includegraphics[width=0.49\textwidth]{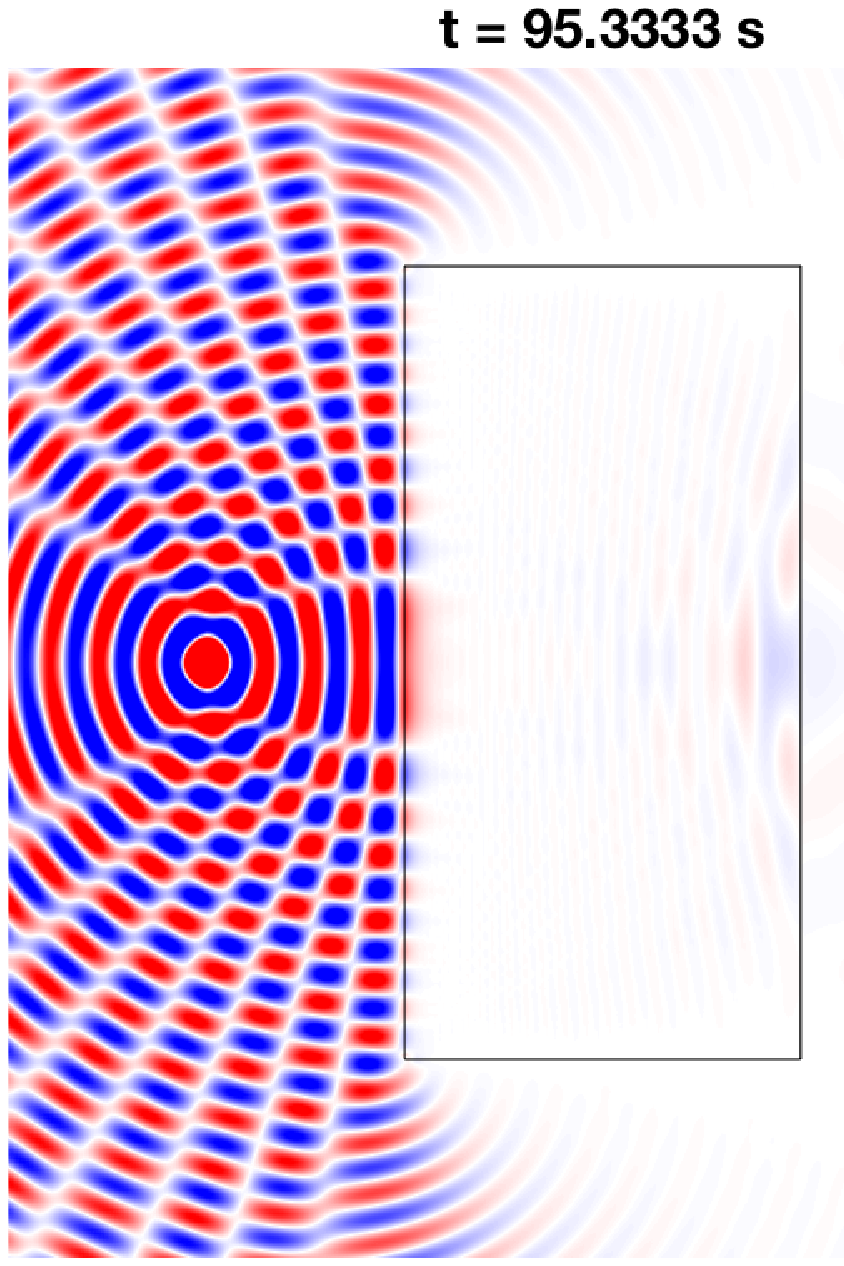}
\hfill
\includegraphics[width=0.49\textwidth]{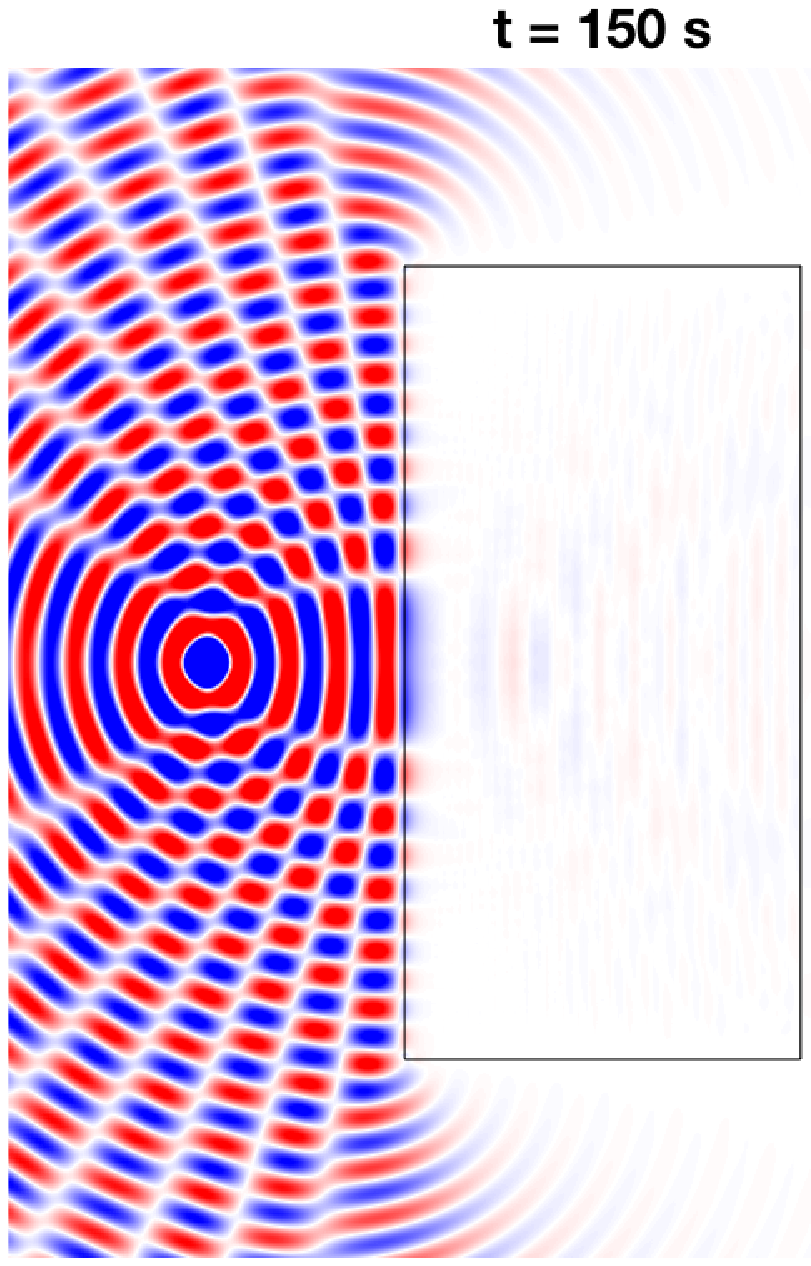}

\label{fig:expe2}
\end{figure}

\begin{figure}[!h]

\includegraphics[width=0.49\textwidth]{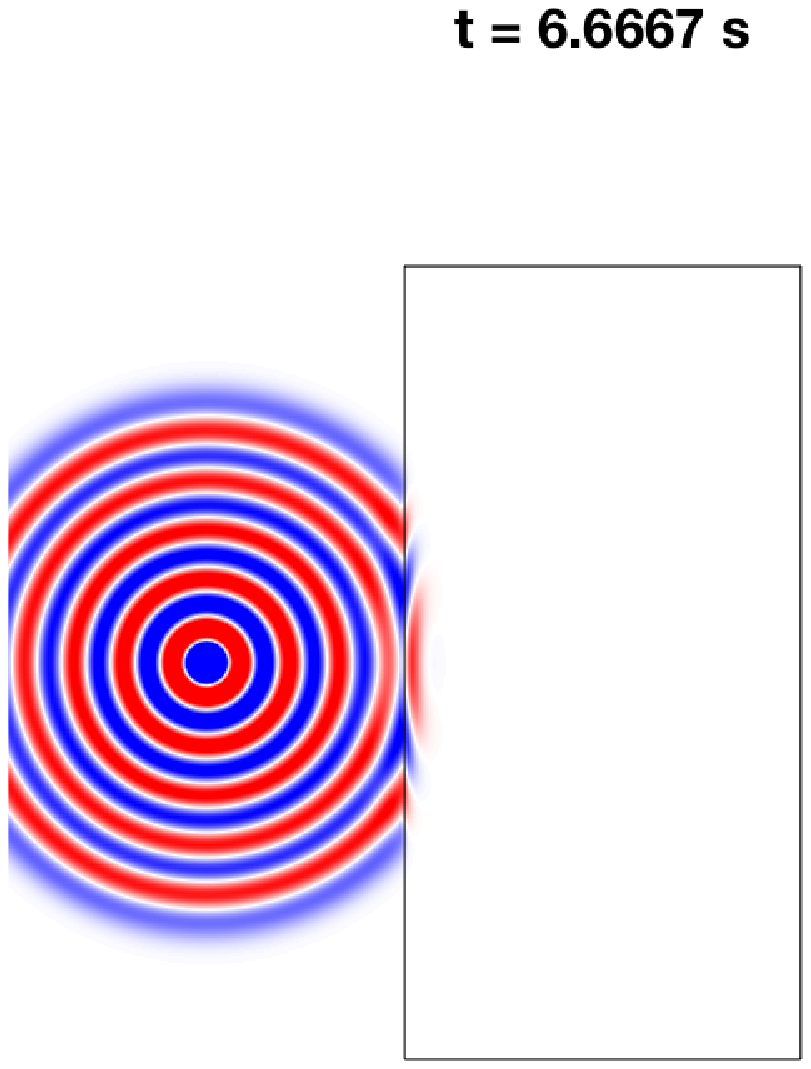}
\hfill
\includegraphics[width=0.49\textwidth]{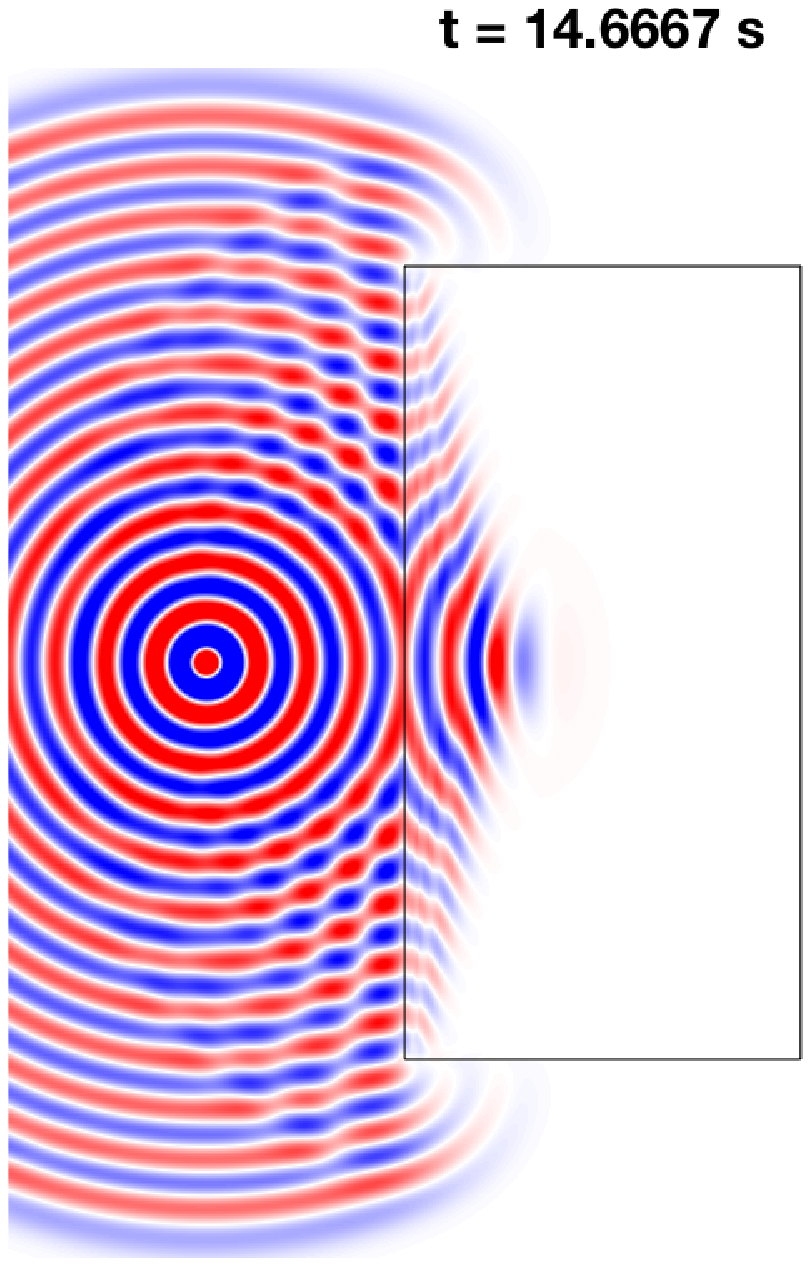}

\includegraphics[width=0.49\textwidth]{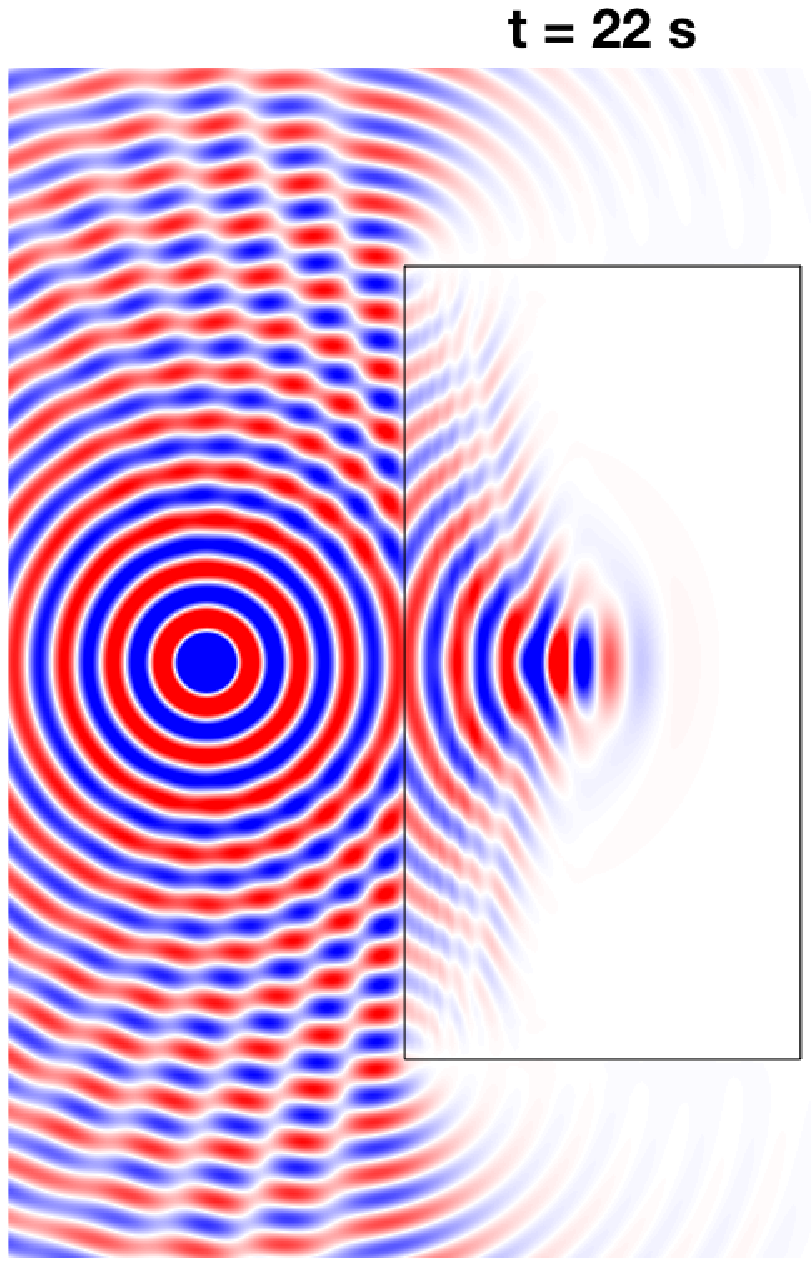}
\hfill
\includegraphics[width=0.49\textwidth]{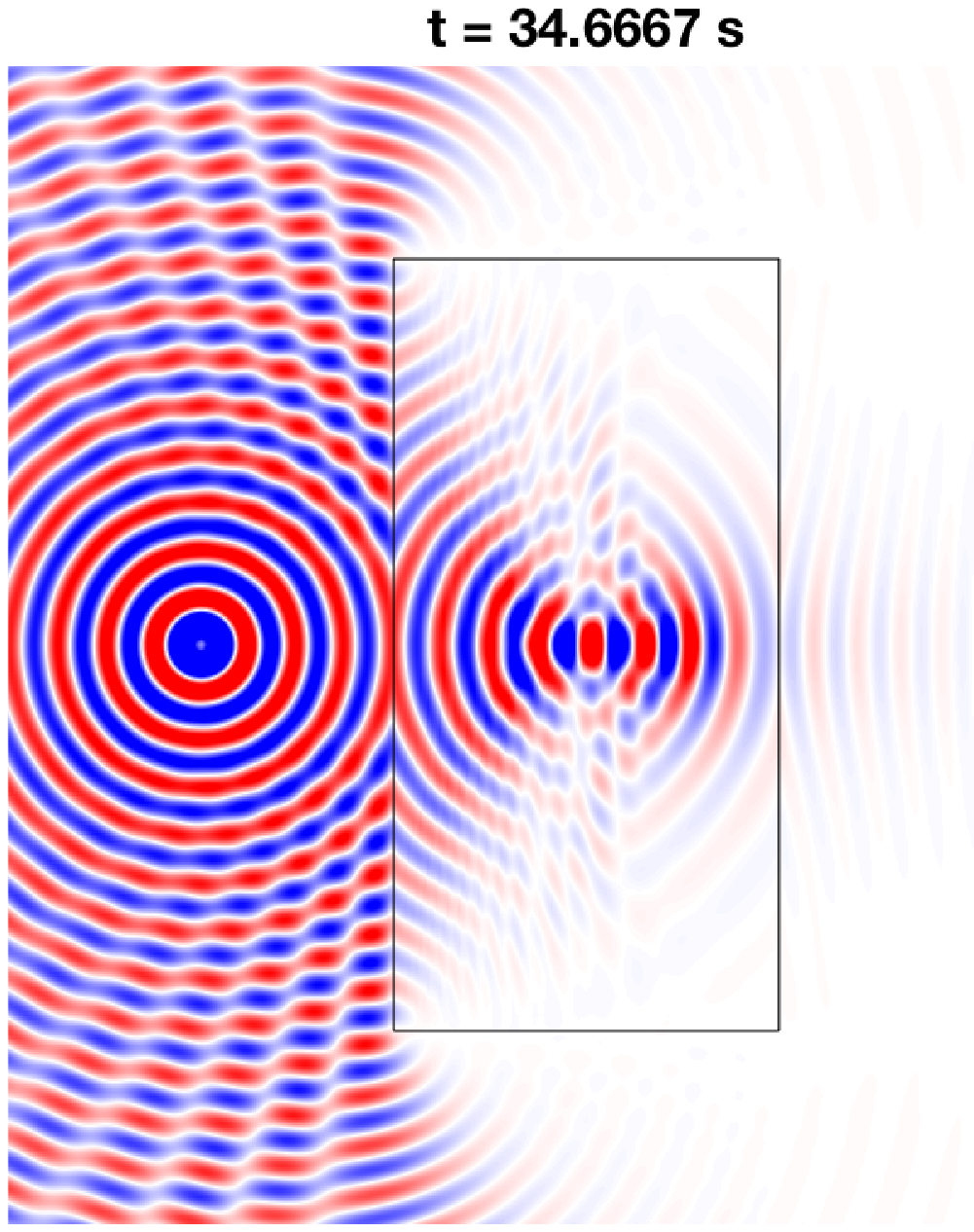}

\includegraphics[width=0.49\textwidth]{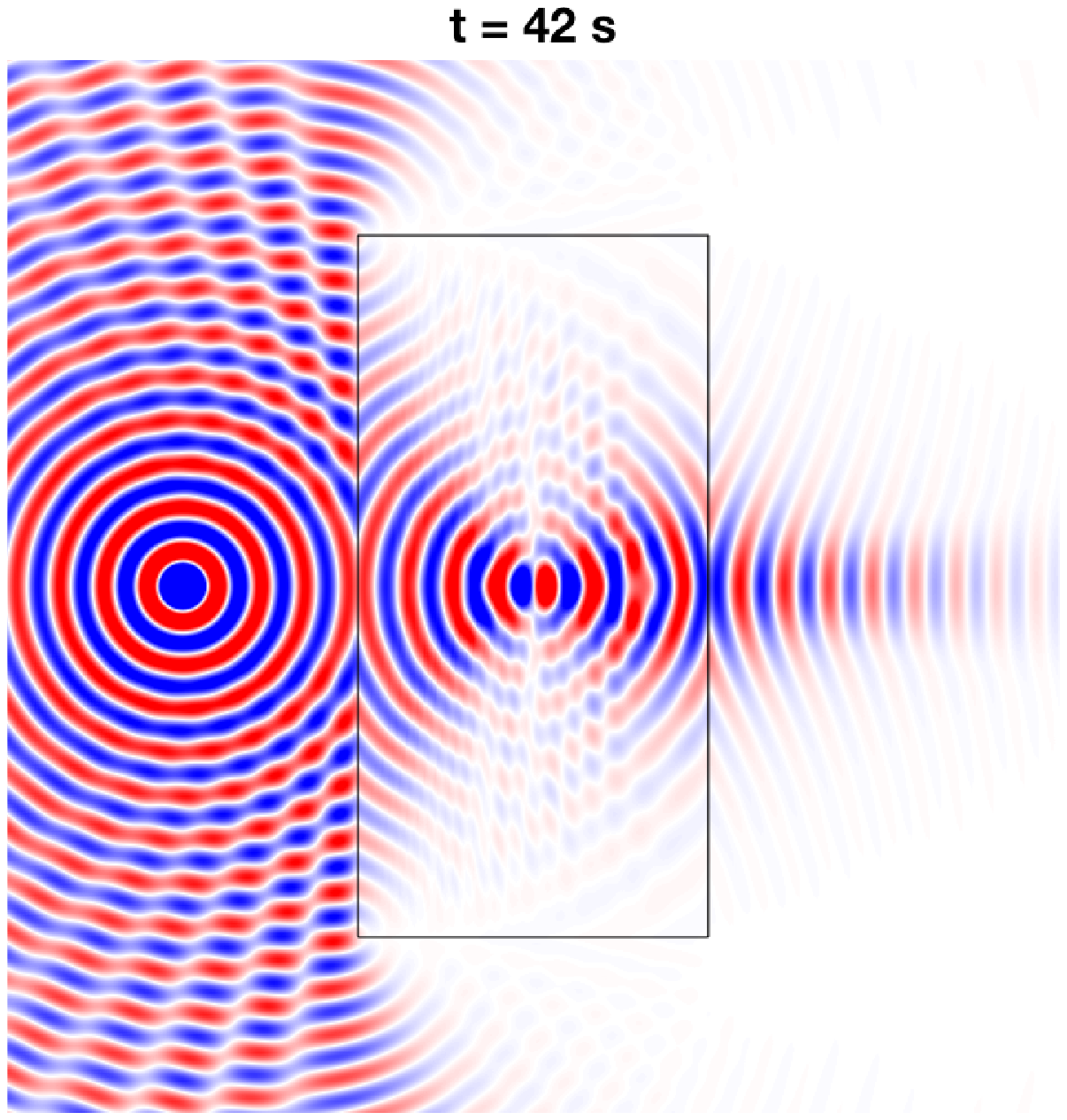}
\hfill
\includegraphics[width=0.49\textwidth]{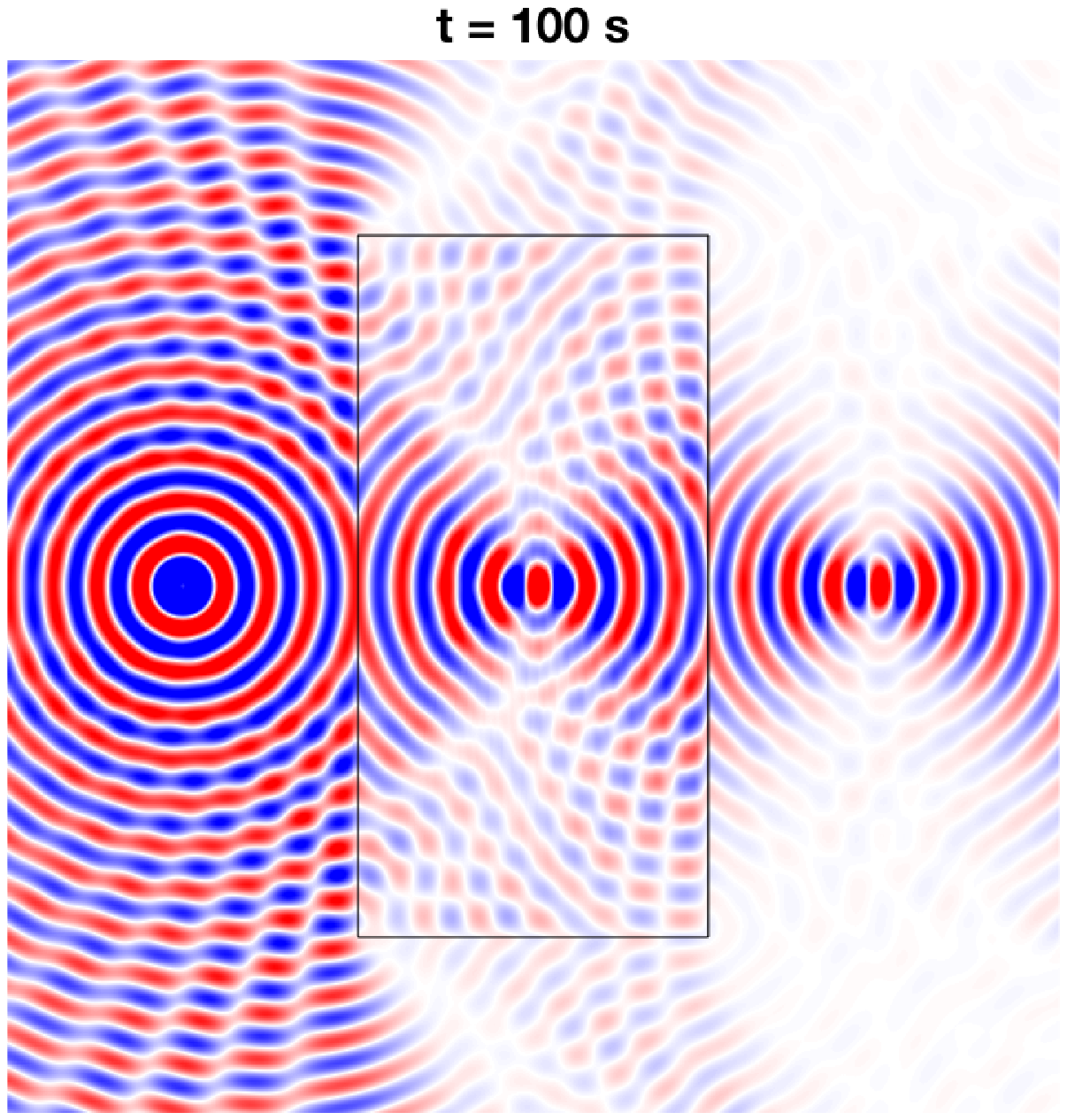}

\caption{Snapshots of $E_3$ at different times for the third experiment.}
\label{fig:expe3}
\end{figure}


\end{document}